\documentclass[a4paper]{article}

\usepackage[T1]{fontenc} % determines what kind of font is used, needed for properly displaying some specific signs and accented characters 
 
\usepackage{calc}               % Simple computations with LaTeX variables

\usepackage{graphicx}           % Standard graphics package

\usepackage{amsfonts,amsmath,amssymb,amsthm}
\def\={=&\:}
\def\+{+&\:}
\def\-{-&\:}

\newcommand{\nn}{\nonumber}

\usepackage{calc,amsmath}
\newlength\dlf

\def\sheaf{\mathcal} % \newcommand wouldn't work here since there is already such command in one of the packages involved
\def\O{\mathcal{O}} % \newcommand wouldn't work here since there is already such command in one of the packages involved

\usepackage{txfonts} % allows for \ointclockwise and \ointctrclockwise

\newcommand{\T}{\langle\mathbf{T}\rangle}
\newcommand{\1}{\langle\mathbf{1}\rangle}
\newcommand{\alphaflat}{\alpha_{\text{flat}}}
\newcommand{\alphasing}{\alpha_{\text{sing.}}}

\newcommand{\oneflat}{\langle\mathbf{1}\rangle_{\text{flat}}}
\newcommand{\onesing}{\langle\mathbf{1}\rangle_{\text{sing.}}}

\newcommand{\ranglesing}{\rangle_{\text{sing}}}
\def\C{\mathbb{C}}
\def\CP{\mathbb{P}_{\C}}

\def\H{\mathbb{H}}
\def\N{\mathbb{N}}

\def\Q{\mathbb{Q}}

\def\R{\mathbb{R}}
\def\RP{\mathbb{P}_{\R}}
\def\Z{\mathbb{Z}}

\newcommand{\tr}{\text{tr}}

\newcommand{\rechts}{\:\rightarrow\:}

\usepackage{graphicx} %for definition of mapsfrom which is the mapsto vector reversed

\def\i{i\:}

\newcommand{\eps}{\varepsilon}

\newcommand{\cyc}{\text{cyclic}}

\newtheorem{theorem}{Theorem}
\newtheorem{'theorem'}{''Theorem``}
\newtheorem{lemma}[theorem]{Lemma}
\newtheorem{proposition}{Propos.}

\newtheorem{remark}[proposition]{Remark}

\newtheorem{definition and theorem}{Definition an Theorem}
\theoremstyle{plain}
\newtheorem*{remark*}{Remark}
\newtheorem*{definition*}{Definition}
\newtheorem*{example*}{Example}

\def\A{\mathbf{A}}
\def\Asing{\A_{\text{sing.}}}
\def\Aflat{\A_{\text{flat}}}
\def\alphasing{\alpha_{\text{sing.}}}
\def\alphaflat{\alpha_{\text{flat}}}
\def\D{\mathfrak{D}}

\begin{document}

\title{An algebraic approach to minimal models in CFTs}
\author{Marianne Leitner$^*$\\
Dublin Institute for Advanced Studies,\\
School of Theoretical Physics,\\
10 Burlington Road, Dublin 4, Ireland\\
$^*$\textit{leitner@stp.dias.ie}
}

\maketitle

\begin{abstract}
CFTs are naturally defined on Riemann surfaces. The rational
ones can be solved using methods from algebraic geometry. One particular
feature is the covariance of the partition function under the mapping
class group. In genus $g=1$, this yields modular forms, 
which can be linked to ordinary differential equations of hypergeometric type
with algebraic solutions.
\end{abstract}

% \tableofcontents

\pagebreak

\section{Introduction}

This is the second in a sequence of three papers on a mathematical approach to Conformal Field Theory (CFT) on compact Riemann surfaces,
and it covers the second part of the author's PhD thesis in Mathematics \cite{L:PhD14}.
In the first part of the thesis, a working definition of rational CFTs on general Riemann surfaces has been given.
For the $(2,5)$ minimal model over compact Riemann surfaces, 
explicit formulae for computing $N$-point functions $\langle\phi_1\ldots\phi_N\rangle$ of holomorphic fields 
have been established for small positive values of $N$.
$N$-point functions for higher $N$ are obtained by recursion.
For $N=0$, one has the identity field $\pmb{1}$ and the partition function $\langle\pmb{1}\rangle$
whose computation requires different methods.
There is no dependence on position, but it depends on the conformal structure of the surface.
Indeed, it satisfies a system of differential equations w.r.t.\ the moduli of the Riemann surface.
For the minimal models, the vector space of solutions is finite dimensional.

The present paper is devoted to compact Riemann surfaces of genus $g=1$.
Such surface can be described as a quotient $\C/\Lambda$,  
with a lattice $\Lambda$ generated over $\Z$ by $1$ and $\tau$ with $\tau\in\mathfrak{h}$,
the complex upper half plane.
The latter is the universal cover of the moduli space $\mathcal{M}_1$ of all possible conformal structures on the $g=1$ surface,
which is known as the Teichm\" uller space. One has $\mathcal{M}_1=SL(2,\Z)\setminus\mathfrak{h}$.
Meromorphic functions on finite covers of $\mathcal{M}_1$ are called (weakly) modular.
They can be described as functions on $\mathfrak{h}$ which are invariant under a subgroup of $SL(2,\Z)$
of finite index. 

Maps in the full modular group $SL(2,\Z)$ preserve the standard lattice $\Z^2$ together with its orientation
and so descend to self-homeomorphisms of the torus. 
Inversely, every self-homeomorphism of the torus is isotopic to such a map. 
A modular function is a function on the space $\mathcal{L}$ of all lattices in $\C$
satisfying \cite{Z:1-2-3}
\begin{displaymath}%\label{equation defining modular functions}
f(\lambda\Lambda)
=f(\Lambda)\:,\quad\forall\Lambda\in\mathcal{L},\:\lambda\in\C^*\:. 
\end{displaymath}
$\mathcal{L}$ can be viewed as the space of all tori with a flat metric.

Conformal field theories on the torus provide many interesting modular functions, and modular forms.
(The latter transform as $f(\lambda\Lambda)=\lambda^{-k}f(\Lambda)$
for some $k\in\Z$ which is specific to $f$, called the weight of $f$.)

For the $(2,5)$ minimal model, we shall derive the second order ordinary differential equation for the $g=1$ partition function 
that allows to compute all $N$-point functions of holomorphic fields.
It is shown that our approach reproduces the known result.

Much of the mathematical foundations of rational CFT will be provided by the joint paper with W.\ Nahm, 
whose main feature are the ODEs for the higher genus partition functions.

% For an important class of CFTs (the minimal models), 
% the zero-point functions $\1$ will turn out to solve a linear differential equation
% so that $\1$ can be computed for arbitrary hyperelliptic Riemann surfaces.
% Since $\1$ is algebraic (namely a meromorphic function on a finite covering of the moduli space),
% it is clear a priori that the equation can not be solved numerically only, but actually analytically.

\section{Introduction to modular dependence}

Given $q=e^{2\pi\i\tau}$ and $\tau\in\mathfrak{h}$, 
let 
\begin{displaymath}
\Sigma:=\{z\in\C|\:|q|\leq z\leq 1\}/\{z\sim qz\}
\:. 
\end{displaymath}
$\Sigma$ is a torus.
A \textsl{character} on $\Sigma$ is given by
\begin{align*}
\langle\pmb{1}\rangle_{\Sigma,i}
\=\underset{\{\varphi_j\}_j\:\text{basis of}\:F_i}{\sum_{\varphi_j}} q^{h(\varphi_j)}\:.
\end{align*}
Here $F_i$ is a fiber of the bundle of fields $\sheaf{F}$ in a rational CFT on $\Sigma$.
For $\varphi_j\in F_i$, $h(\varphi_j)$ is the conformal weight of $\varphi_j$. 
%and $\langle\pmb{1}\rangle_{\Sigma,i}$ is in particular a $0$-point function $\1$, in the sense of \cite{LN:2017}.
%On the other hand, 
$\langle\pmb{1}\rangle_{\Sigma,i}$ is a modular function of $\tau$ \cite{Nahm}.
A \textsl{modular function on} a discrete subgroup $\Gamma$ of $\Gamma_1=SL(2,\Z)$ 
is a $\Gamma$-invariant meromorphic function $f:\mathfrak{h}\rechts\C$ 
with at most exponential growth towards the boundary \cite{Z:1-2-3}.
For $N\geq 1$, the \textsl{principal conguence subgroup} is the group $\Gamma(N)$ such that the short sequence
\begin{displaymath}
1\rechts\Gamma(N)\hookrightarrow\:\Gamma_1\:\overset{\pi_N}{\longrightarrow}SL(2,\Z/N\Z)\rechts 1
\end{displaymath}
is exact, where $\pi_N$ is map given by reduction modulo $N$.
A function that is modular on $\Gamma(N)$ is said to be \textsl{of level $N$}.
Let $\zeta_N=e^{\frac{2\pi\i}{N}}$ be the $N$-th root of unity with cyclotomic field $\Q(\zeta_N)$.
Let $F_N$ be the field of modular functions $f$ of level $N$
which have a Fourier expansion 
\begin{align}\label{Fourier expansion of modular function}
f(\tau)
=\sum_{n\geq-n_0}a_nq^\frac{n}{N}\:,
\quad
q=e^{2\pi\i\tau}\:, 
\end{align}
with $a_n\in\Q(\zeta_N)$, $\forall n$. 
%(Here $n_0\in\N$ is the order of the pole of $f$ at $i\infty$.)
The Ramanujan continued fraction
\begin{align}\label{Ramanujan continued fraction}
r(\tau)
:=q^{1/5}\frac{1}{1+\frac{q}{1+\frac{q^2}{1+...}}}
%=q^{\frac{1}{5}}\left(1-q+q^2-q^4+q^5-q^6+\ldots\right)
\:
\end{align}
which converges for $\tau\in\mathfrak{h}$, 
is an element (actually a generator) of $F_5$ \cite{CJMV-Z:2007}. 
$r$ is algebraic over $F_1$ (cf.\ Section \ref{section: The property of being algebraic}) 
which is generated over $\Q$ by the \textsl{modular $j$-function},
\begin{displaymath}
j(\tau)
=12^3\frac{g_2^3}{g_2^3-27g_3^2}
\:. 
\end{displaymath}
$j$ is associated to the elliptic curve with the affine equation
\begin{displaymath}
\Sigma:\quad y^2=4x^3-g_2x-g_3\:,\quad\text{with}\quad g_2^3-27g_3^2\not=0\:. 
\end{displaymath}
Here $g_{k}$ for $k=2,3$ are (specific) modular forms of weight $2k$,\footnote{As mentioned earlier, 
a modular form of weight $2k$ transforms as $f(\lambda\Lambda)=\lambda^{-2k}f(\Lambda)$ for any $\lambda\in\C^*$.}
so that $j$ is indeed a function of the respective modulus only
(the quotient $\tau=\omega_2/\omega_1$ for the lattice $\Lambda=\Z.\omega_1+\Z.\omega_2$),
or rather its orbit under $\Gamma_1$ (since we are free to change the basis $(\omega_1,\omega_2)$ for $\Lambda$).
%$\C/\Lambda$ and $\C/\Lambda'$ are isomorphic if and only if $\Lambda'=\lambda\Lambda$ for some $\lambda\in\C^*$
%\cite{S:1973}.
In terms of the modulus, 
a \textsl{modular form of weight} $2k$ \textsl{on} $\Gamma$ is a holomorphic function $g:\mathfrak{h}\rechts\C$ 
with subexponential growth towards the boundary \cite{Z:1-2-3} 
such that $g(\tau)\:(d\tau)^{2k}$ is $\Gamma$-invariant \cite{S:1973}.
A modular form on $\Gamma_1$ allows a Fourier expansion of the form (\ref{Fourier expansion of modular function})
with $n_0\geq 0$.

Another way to approach modular functions is in terms of the differential equations they satisfy. 
The derivative of a modular function is a modular form of weight two, 
and higher derivatives give rise to \textsl{quasi-modular forms},
which we shall also deal with though they are not themselves of primary interest to us.

Geometrically, the conformal structure on the surface 
\begin{align}\label{eq: g=1 Riemann surface}
\Sigma:\quad y^2=4(x-X_1)(x-X_2)(x-X_3)\:,\quad x\in\CP^1\:, 
\end{align}
is determined by the quadrupel 
$(X_1,X_2,X_3,\infty)$ of its ramification points, 
and we can change this structure by varying the position of $X_1,X_2,X_3$ infinitesimally. 
In this picture, the boundary of the moduli space is approached by letting two ramification points 
in the quadrupel run together \cite{FS:1987}.

When changing positions 
we may keep track of the branch points
to obtain a simply connected space \cite{Don:2011}.
Thus a third way to describe modularity of the characters is by means of a subgroup of the braid group $B_3$ of $3$ strands.
The latter is the universal central extension of the quotient group $\overline{\Gamma}_1=\Gamma_1/\{\pm\pmb{I}_2\}$, 
so that we come full circle. 
% The braid group $B_3$ of three strands is the universal central extension of $\overline{\Gamma}_1$. 
% We have
% \begin{align*}
% \frac{B_3}{\langle c\rangle}\cong PSL(2,\Z) 
% \end{align*}
% where $c$ is the central element 
% \begin{align*}
% c=a^2=b^3 
% \end{align*}
% for $a=\Sigma\sigma_2\sigma_3$ and $b=\Sigma\sigma_2$. Here $\sigma_i$ moves the $i+1$st strand over the $i$'th strand. 
% $a$ and $b$ correspond to $S$ and $ST$, respectively, in $PSL(2,\Z)$, where $S$ and $T$ are the generators of the modular group. 
% $B_3$ is also isomorphic to the mapping class group of a punctured disk with $3$ punctures.

Suppose $\Sigma=\C/\Lambda$ where $\Lambda=(\Z.1+\Z.i\beta)$ with $\beta\in\R$.
Thus the fundamental domain is a rectangle in the $(x^0,x^1)$ plane
with length $\Delta x^0=1$ and width $\Delta x^1=\beta$.
The dependence of $\langle\pmb{1}\rangle_{\Sigma}$ on the modulus $i\beta$ follows from the identity
\begin{align*}
\langle\pmb{1}\rangle_{\Sigma,0}
=\tr\:e^{-H\beta}\:,
\quad
H
=\int T^{00}dx^0, 
\end{align*}
where $T^{00}$ is a real component of the Virasoro field.\footnote
{Any dynamical quantum field theory has an energy-momentum tensor $T_{\mu\nu}$ s.t.\
$T_{\mu\nu}dx^{\mu}dx^{\nu}$ defines a quadratic differential. 
In particular  it transforms homogeneously under coordinate changes. 
For coordinates $z=x^0+ix^1$ and $\bar{z}=x^0-ix^1$, we have \cite{Blum:2009}
\begin{displaymath}
T_{zz}=\frac{1}{4}(T_{00}-2iT_{10}-T_{11})\:.
\end{displaymath}
For a discussion of the relation with the Virasoro field $T(z)$ addressed below, cf.\ \cite{LN:2017}.} 
We may regard $\langle\pmb{1}\rangle_{\Sigma,0}$ as the $0$-point function $\1$ on $\Sigma$.

Stretching $\beta\mapsto(1+\epsilon)\beta$ changes the Euclidean metric $G_{\mu\nu}$ ($\mu,\nu=0,1$) according to
\begin{displaymath}
(ds)^2
\:\mapsto\:(ds)^2+2\epsilon(d x^1)^2+O(\epsilon^2)\:. 
\end{displaymath}
Thus $d G_{11}=2\frac{d\beta}{\beta}$, and 
\begin{align}\label{the variation of the zero-point function for the case g=1}
d\1
=-\tr(H d\beta\:e^{-H\beta})
\=-\frac{d G_{11}}{2}\:\left(\int\langle T^{00}\rangle dx^0\right)\:\beta\nn\\
\=-\frac{d G_{11}}{2}\:\iint\langle T^{00}\rangle dx^0dx^1\:. 
\end{align} 
The fact that $\int\langle T^{00}\rangle dx^0$ does not depend on $x^1$ follows from the conservation law
$\partial_{\mu}T^{\mu\nu}=0$:
\begin{align*}
\frac{d}{d x^1}\oint\langle T^{00}\rangle\:dx^0
=\oint\partial_1\langle T^{00}\rangle\:dx^0 
=-\oint\partial_0\langle T^{10}\rangle\:dx^0 
=0\:,
\end{align*}
using Stokes' Theorem.

We argue that on $S^1\times S^1_{\beta/(2\pi)}$ 
(where $S^1_{\beta/(2\pi)}$ is the circle of perimeter $\beta$),
states (in the sense of \cite{L:2013})
are thermal states on the VOA. 

When $g>1$, equation (\ref{the variation of the zero-point function for the case g=1}) generalises to
\begin{displaymath}%\label{variation of the zero-point function}
d\1 
=-\frac{1}{2}\iint d G_{\mu\nu}\:\langle T^{\mu\nu}\rangle\:\sqrt{G}\:dx^0\wedge dx^1\:.
\end{displaymath}
Here $G:=|\det G_{\mu\nu}|$, 
and $dvol_2=\sqrt{G}\:dx^0\wedge dx^1$ is the volume form which is invariant under base change.\footnote{
The change to complex coordinates is more intricate: 
We have $dx^0\wedge dx^1=iG_{z\bar{z}}\:dz\wedge d\bar{z}$ with $G_{z\bar{z}}=\frac{1}{2}$.}  
The normalisation is in agreement with eq.\ (\ref{the variation of the zero-point function for the case g=1})
(see also \cite{DiFranc:1997}, eq.\ (5.140) on p.\ 139). 

Methods that make use of the flat metric do not carry over to surfaces of higher genus.
We may choose a specific metric of prescribed constant curvature
to obtain mathematically correct but cumbersome formulae.
Alternatively, 
we consider quotients of $N$-point functions over $\1$ only (as done in \cite{EO:1987}) 
so that the dependence on the specific metric drops out.  
Yet we suggest to use a singular metric that is adapted to the specific problem \cite{LN:2017}.
On $\Sigma$, this metric is the lift of the flat metric 
\begin{displaymath}
|dz|^2
\quad
\text{on}\:\CP^1\setminus\{X_1,X_2,X_3,\infty\}\:,
\end{displaymath}
and has all curvature concentrated in the ramification points.
The $0$-point function on this metric surface is obtained through a regularisation procedure
and will be denoted $\onesing$ to distinguish it from the $0$-point function on the flat torus $(\Sigma,|dz|^2)$, which we denote by $\oneflat$.
Unless otherwise stated, all state-dependent objects are understood to refer to the singular metric on $\Sigma$.

\begin{theorem}%\label{Thm: Main theorem}\cite{LN:2017}
Let $\Sigma$ be defined by eq.\ (\ref{eq: g=1 Riemann surface}).
We equip $\Sigma$ with the metric which is the lift of the flat metric on $\C$ to its double cover.
%Let $\langle\:\ranglesing$ be a corresponding state on $\Sigma$.
Define a deformation of the conformal structure by 
\begin{displaymath}
\xi_j=dX_j\quad\text{for}\quad j=1,2,3\:.
\end{displaymath}
Let $\varphi,\ldots$ be holomorphic fields on $\Sigma$.
For $j=1,2,3$, let $(U_j,z)$ be a chart on $\Sigma$ containing the point $X_j$ but no position of one of $\varphi,\ldots$.
We have
\begin{displaymath}%\label{main formula, but for alpha=1 and N=1}
d\langle\varphi\ldots\ranglesing
=\sum_{j=1}^{n}\left(\frac{1}{2\pi\i}\ointctrclockwise_{\gamma_j}\langle T(z)\varphi\ldots\ranglesing\:dz\right)\:\xi_j\:,
\end{displaymath} 
where $\gamma_j$ is a closed path around $X_j$ contained in $U_j$. 
\end{theorem}

\section{Differential equations for characters in $(2,\nu)$-minimal models}\label{Chapter: DE for characters}

%The use of the term character is specific to $g=1$.
%Most of the results of the present chapter can be found in the literature.

\subsection{Review of the MLDE for the characters of the $(2,5)$ minimal model}\label{Section: differential equations for the (2,5) minimal model}

The character $\1$ of any CFT on the torus $\Sigma$ solves the ODE \cite{EO:1987}
\begin{align}\label{ODE for zero-point function in z coordinate in the (2,5) minimal model}
\frac{1}{2\pi i}\frac{d}{d\tau}\1
&=\oint\langle T(z)\rangle\:\frac{dz}{(2\pi i)^2}
=\frac{1}{(2\pi i)^2}\T\:.
\end{align}
Here the contour integral is along the real period, and $\oint dz=1$.
$\T$, while constant in position, is a modular form of weight two in the modulus.\footnote{$\langle 1\rangle$, $\langle T\rangle$ (or later $\text{A}$)
are parameters of central importance to this exposition. 
For better readibility, they appear in bold print ($\1$ and $\T$, or $\A$) 
throughout.}
The Virasoro field generates the variation of the conformal structure \cite{EO:1987}.
In the $(2,5)$ minimal model,
we find 
%by eqs (\ref{2-pt function}) and (\ref{C}) in Part I, 
\begin{align}\label{ODE for one-point function}
\frac{1}{2\pi i}\frac{d}{d\tau}\T
=\oint\langle T(w)T(z)\rangle\:\frac{dz}{(2\pi i)^2}
=\left(\frac{1}{6}E_2\T-\frac{11}{900}\pi^2E_4\1\right)\:.
\end{align}
Here $E_2$ is the quasimodular Eisenstein series of weight $2$, 
which enters the equation by means of the identity
\begin{displaymath}
\int_0^1\wp(z-w|\tau)\:dz
=-\frac{\pi^2}{3}E_2(\tau).
\end{displaymath} 
In terms of the Serre-derivative operator $\D=qd/dq-(\ell/12)E_2(q)$ (defined on modular forms of weight $\ell\in\R$),
the first order ODEs 
(\ref{ODE for zero-point function in z coordinate in the (2,5) minimal model}) and (\ref{ODE for one-point function})
combine to give the second order ODE \cite{MMS:1988,Z-K:1997}
\begin{align}\label{eq: 2nd order ODE for <1> in the (2,5) minimal model}
\D^2\1
=\frac{11}{3600}\:E_4\1\:. 
\end{align}
The two solutions are the well-known Rogers-Ramanujan functions \cite{DiFranc:1997}
\begin{equation}\label{defs: Rogers-Ramanujan functions}
\begin{split}
\1_1(q)
=&\:H(q)
:=q^{\frac{11}{60}}\sum_{n\geq 0}\frac{q^{n^2+n}}{(q;q)_n}
=q^{\frac{11}{60}}\left(1+q^2+q^3+q^4+q^5+2q^6+\ldots\right)
\:,\\
\1_2(q)
=&\:G(q)
:=q^{-\frac{1}{60}}\sum_{n\geq 0}\frac{q^{n^2}}{(q;q)_n}
=q^{-\frac{1}{60}}\left(1+q+q^2+q^3+2q^4+\ldots\right)
\:.  
\end{split}
\end{equation}
($(q;q)_{n}:=\prod_{k=1}^n(1-q^k)$ is the $q$-Pochhammer symbol) which are named after the famous Rogers-Ramanujan identities
\begin{align}\label{eq: Rogers-Ramanujan identities}
q^{-\frac{11}{60}}\1_1=\prod_{n=\pm 2\:\text{mod}\:5}(1-q^n)^{-1}\:,
\quad 
q^{\frac{1}{60}}\1_2=\prod_{n=\pm 1\:\text{mod}\:5}(1-q^n)^{-1}\:.
\end{align}
$q^{-\frac{11}{60}}\1_1$ provides the generating function for the partition
which to a given holomorphic dimension $h\geq 0$ returns the number of linearly independent holomorphic fields 
present in the $(2,5)$ minimal model.
This number is subject to the constraint $\partial^2T\propto N_0(T,T)$
%Recall that this number is subject to the constraint $\partial^2T\propto N_0(T,T)$, 
%eq.\ (\ref{Phi when h=4 has dim=1}) in Part I.

\begin{table}[ht]%begin table here
\centering
\begin{tabular}{|l|l|l|l|l|l|l|l|l}
\hline
$h$&$0$&$1$&$2$&$3$&$4$&$5$&$6$\\ 
\hline
basis of $F(h)$&$1$&$-$&$T$&$\partial T$&$\partial^2T$&$\partial^3T$&$\partial^4T$\\
&&&&&&&$N_0(T,\partial^2 T)$\\
\hline
$\dim F(h)$&$1$&$0$&$1$&$1$&$1$&$1$&$2$\\
\hline
\end{tabular}\\
\caption{Holomorphic fields of dimension $h$ in the $(2,5)$ minimal model}
%\label{table: holomorphic fields of dimension h in the (2,5) MM}
\end{table}

There is a similar combinatorical interpretation for the second Rogers-Ramanujan identity.
It involves non-holomorphic fields, however, which we disregard in this paper.

\subsection{MLDE for the characters in $(2,\nu)$ minimal models}%\label{Section: Modular ODEs for the characters in (2,nu) minimal models}

Sorting out the algebraic equations to describe the characters of the $(2,\nu)$ minimal model becomes tedious for $\nu>5$.
In contrast, the Serre derivative is a manageable tool for encoding them in a compact way 
\cite{MMS:1988}. 
Since the characters are algebraic, 
the corresponding differential equations can not be solved numerically only, but actually analytically.
We are interested in the fact that the coefficient of the respective highest order derivative can be normalised to one 
and all other coefficients are holomorphic in the modulus.

The $(2,\nu)$ minimal model, where $\nu\geq 3$ is odd, has
\begin{displaymath}
M=\frac{\nu-1}{2} 
\end{displaymath}
characters \cite{DiFranc:1997}.
They are parametrised by the sequence
\begin{align}\label{kappa in (2,nu) minimal model}
\kappa_s
%=h_{1,s}-\frac{c}{24}
=\frac{(\nu-2s)^2}{8\nu}-\frac{1}{24}\:,
\quad s=1,\dots,M\:.
\end{align}
The character corresponding to $\kappa_s$ is 
\begin{displaymath}
\1_{\kappa_s}
=f_{A,B,C}
\:,
\end{displaymath}
where $f_{A,B,C}$ is the $q$-hypergeometric series
\begin{displaymath}
f_{A,B,C}
:=\sum_{\textbf{n}\in(\N_0)^r}\frac{q^{\frac{1}{2}\textbf{n}^tA\textbf{n}+\textbf{B}^t\textbf{n}+C}}{(q;q)_{\textbf{n}}}
\:, 
\end{displaymath}
with $r=(\nu-3)/2$ being the rank, and 
\begin{displaymath}
A=\mathcal{C}(T_r)^{-1}\in\Q^{r\times r},
\quad
\textbf{B}\in\Q^r\:,
\quad
C=\kappa_s
\:. 
\end{displaymath} 
Here $\mathcal{C}(T_r)$ denotes the Cartan matrix of the tadpole diagram $T_r$.
The latter is obtained from the Dynkin diagram of $A_{2r}$ by folding according to its $\Z_2$ symmetry,
and for $i,j\in\{1,\ldots,r\}$,
\begin{displaymath}
A_{ij}=2\delta_{ij}-\sharp\{\text{links between nodes $i$ and $j$}\}
\:.
\end{displaymath}
For example, in the $(2,7)$ minimal model, 
$A=\mathcal{C}(T_2)^{-1}=\begin{pmatrix}
      2&-1\\
      -1&1
      \end{pmatrix}$.
      
It turns out that $\1_{\kappa_s}$ satisfies an $M$th order ODE \cite{MMS:1988}. 
Given $M$ differentiable functions $f_1,\ldots,f_M$ there always exists an ODE having these as solutions. 
Consider the Wronskian determinant
\begin{align*}
\det
\begin{pmatrix}
f&\D^1f&\ldots&\D^Mf\\
f_1&\D^1f_1&\ldots&\D^Mf_1\\
\ldots&\ldots&\ldots&\ldots\\
f_M&\D^1f_M&\ldots&\D^Mf_M 
\end{pmatrix}
=:\sum_{i=0}^Mw_i\:\D^if\:.
\end{align*}
Here for $m\geq 1$, $\D^m$ is the $m$-fold composition of the Serre differential operator,
which maps modular functions into modular forms  of weight $2m$.
For $m=0$ we set $\D^0=1$.

Whenever $f$ equals one of the $f_i$, $1\leq i\leq M$, the determinant is zero, 
so we obtain an ODE in $f$ whose coefficients are Wronskian minors containing $f_1,\ldots,f_M$ and their derivatives only.
These are modular when the $f_1,\ldots,f_M$ and their derivatives are
or when under modular transformation, they transform into linear combinations of one another (as the characters do).

The space of holomorphic modular forms $M(\Gamma_1)$ on $\Gamma_1$ 
has a basis $\Delta^kE_{2m}$, $k=0,1,2,\ldots$, $m=0,2,3,\ldots$.
Here $\Delta$ is the modular discriminant function (of weight $12$), 
and $E_{2m}$ is the holomorphic Eisenstein series of weight $2m$.

\begin{proposition}
For $i=1,\ldots,M$, the characters $f_i$ of the $(2,\nu)$ minimal model satisfy an ODE
\begin{align}\label{eq: differential equation for the characters in the (2,nu) minimal model}
\mathcal{L}^{(2,\nu)}f_i
=0\:, 
\end{align}
where $\mathcal{L}^{(2,\nu)}$ is the differential operator 
\begin{displaymath}
\mathcal{L}^{(2,\nu)}
:=\D^M
+\sum_{m=2}^{M}\Omega_{2m}\D^{M-m}
\:.
\end{displaymath}
Here the $\Omega_{2m}$ are holomorphic modular forms on $\Gamma_1$ of weight $2m$.
For $2\leq m\leq5$, there exist $\alpha_{M-m}\in\Q$ such that
\begin{displaymath}
\Omega_{2m}
=\alpha_{M-m}E_{2m}
\:.
\end{displaymath}
Moreover, 
\begin{displaymath}
\Omega_{12}
=\alpha_0E_{12}+\alpha_0^{(\text{cusp})}\Delta
\:,
\end{displaymath}
for some $\alpha_0,\alpha_0^{(\text{cusp})}$.
For $3\leq \nu\leq 13$, the nonzero coefficients are given by Table \ref{table: coefficients of the Serre differential operator}.

\begin{table}[ht]%begin table here
%\centering
\begin{tabular}{|l|l|l|l|l|l|l|}
\hline
$\nu$&$3$&$5$&$7$&$9$&$11$&$13$\\ 
\hline
$M$&$1$&$2$&$3$&$4$&$5$&$6$\\
\hline
&&&&&&\\
$\kappa_M$&$0$&$-\frac{1}{60}$&$-\frac{1}{42}$&$-\frac{1}{36}$&$-\frac{1}{33}$&$-\frac{5}{156}$\\
&&&&&&\\
\hline
% &&&&&&\\
% $\alpha_{M}$&$1$&$1$&$1$&$1$&$1$&$1$\\
&&&&&&\\
$\alpha_{M-2}$&&$-\frac{11}{60^2}$&$-\frac{5\cdot7}{42^2}$&$-\frac{2\cdot3\cdot13}{36^2}$&$-\frac{11\cdot53}{2^2\cdot33^2}$&$-\frac{7\cdot13\cdot67}{156^2}$\\
&&&&&&\\
$\alpha_{M-3}$&&&$\:\:\frac{5\cdot17}{42^3}$&$\:\:\frac{2^3\cdot53}{36^3}$&$\:\:\frac{3\cdot5\cdot11\cdot59}{2^3\cdot33^3}$&$\:\:\frac{2^3\cdot13\cdot17\cdot193}{156^3}$\\ 
&&&&&&\\
$\alpha_{M-4}$&&&&$-\frac{3\cdot11\cdot23}{36^4}$&$-\frac{11\cdot6151}{2^4\cdot33^4}$&$-\frac{5\cdot11\cdot13\cdot89\cdot127}{156^4}$\\
&&&&&&\\
$\alpha_{M-5}$&&&&&$\:\:\frac{2^4\cdot17\cdot29}{33^5}$&$\:\:\frac{2^3\cdot3\cdot5\cdot13\cdot31\cdot2437}{156^5}$\\
&&&&&&\\
$\alpha_{M-6}$&&&&&&$-\frac{5^4\cdot7^2\cdot23\cdot31\cdot67}{156^6}$\\
&&&&&&\\
\hline
&&&&&&\\
$\alpha_{M-6}^{(\text{cusp})}$&&&&&&$\frac{5^2\cdot7\cdot11\cdot23^2\cdot167}{2^5\cdot3^2\cdot13^4\cdot691}$\\
&&&&&&\\
\hline
\end{tabular}
\caption{The nonzero coefficients in the order $M$ differential operator in the $(2,\nu)$ minimal model. 
$\kappa_M$ is displayed to explain the standard denominators of the $\alpha_m$.}
\label{table: coefficients of the Serre differential operator}
\end{table}
\end{proposition}

\begin{remark}
\begin{enumerate}
\item In the $(2,\nu)$ minimal model, we have $\kappa_M=(3-\nu)/(24\nu)$,
where $\nu|(3-\nu)\:\Leftrightarrow\:\nu=3$. Thus for $\nu>3$, $\kappa_M$ has a factor of $\nu$ in the denominator.
\item The numerators $n_m^{(2,\nu)}$ of $\alpha_m$ in the $(2,\nu)$ minimal model have mostly few factors
in the sense that
\begin{displaymath}
n_m^{(2,\nu)}\approx\text{rad}(n_m^{(2,\nu)})\:,
\end{displaymath}
where the r.h.s.\ is the radical of $n_m^{(2,\nu)}$.
\item
The prime $691$ displayed in the denominator of $\alpha_{M-6}^{(\text{cusp})}$
suggests that Bernoulli numbers are involved in the computations.
This is an artefact of the choice of basis, however.
Using the identity \cite{Z:1-2-3}
\begin{displaymath}%\label{E12 in terms of E4 cube and E6 to the square}
E_{12}
=\frac{1}{691}(441E_4^3+250E_6^2)\:, 
\end{displaymath}
we can write 
\begin{align*}
\Omega_{12}
\=-\frac{5^2\cdot7\cdot23}{2^7\cdot3^5\cdot13^6}
\left(
\frac{53\cdot1069}{2^5}E_4^3
+\frac{6047}{3}E_6^2
\right)
\:.
\end{align*}
\end{enumerate}
\end{remark}

The leading coefficient can be read off from the equation for the singular vector (Lemma 4.3 in \cite{Wang:1993})
and only the specific value of the remaining coefficients in eq.\ (\ref{eq: differential equation for the characters in the (2,nu) minimal model}) seem to be new. 
Rather than setting up a closed formula for $\alpha_m$, 
we shall outline the algorithm to determine these numbers, and leave the actual computation as an easy numerical exercise.    

\begin{proof}[Sketch of the Proof]
We first show that the highest order coefficient $\alpha_M$ of the ODE can be normalised to one.
For every $\kappa_s$ in the list (\ref{kappa in (2,nu) minimal model}) and for $0\leq m\leq M-1$, 
we have
\begin{displaymath}%\label{shape of D to the power m applied to any character}
\D^m\1_{\kappa_s}\propto q^{\kappa_s}(1+O(q))\:. 
\end{displaymath}
Since the $\kappa_s$ are all different, we know that
\begin{displaymath}
w_M\sim\prod_s q^{\kappa_s}\:,
\quad q\:\text{close to zero}\:,
\end{displaymath}
where $w_M$ is the coefficient of $\D^M$ in the Wronskian. 
By construction, $w_M$ has no pole at finite $\tau$.
The number of zeros can be calculated using Cauchy's Theorem \cite{Z:1-2-3}:
Since $\D^m\1$ has weight $2m$, we find
\begin{displaymath}
\text{weight}\:w_M
=2\sum_{\ell=0}^{M-1}\ell
=M(M-1)\:. 
\end{displaymath}
The order of vanishing $\text{ord}_P(w_M)$ of $w_M$ at a point $P\in\Gamma\setminus\mathfrak{h}$ 
depends only on the orbit $\Gamma P$ \cite{Z:1-2-3}.
Denote by $\text{ord}_{\infty}(w_M)$ the order of vanishing of $w_M$ at $\infty$
(i.e.\ the smallest integer $n\geq 0$ such that $a_n\not=0$ in the Fourier expansion for $w_M$).
A finite index subgroup $\Gamma$ of $\Gamma_1$ has the fundamental domain
\begin{align}\label{eq: fundamental domain of finite index subgroup of Gamma1}
\sheaf{F}_{\Gamma}
=\gamma_1\sheaf{F}\cup\ldots\cup\gamma_{[\Gamma_1:\Gamma]}\sheaf{F}
\:
\end{align}
\cite{Gu:1962}.
Thus all orders of vanishing for $\Gamma$ differ from those for $\Gamma_1$ by the same factor. 
Thus \cite[Propos.\ 2 on p.\ 9]{Z:1-2-3} generalises to subgroups $\Gamma\subset\Gamma_1$ and to 
\begin{displaymath}%\label{number of zeros+order of vanishing at cusp equals weight over 12} 
\text{ord}_{\infty}(w_M)+\sum_{P\in\Gamma\setminus\mathfrak{h}}\frac{1}{n_P}\text{ord}_P(w_M)=\frac{M(M-1)}{12}
\:,
\end{displaymath}
where $n_P$ is the order of the stabiliser.
Since
\begin{displaymath}
\text{ord}_{\infty}(w_M)
=\sum_{s=1}^M\kappa_s
=\frac{M(M-1)}{12}\:, 
\end{displaymath}
we have $\text{ord}_P(w_M)=0$ for $P\in\Gamma\setminus\mathfrak{h}$.
Thus we can divide by $w_M$ to obtain 
\begin{displaymath}
\sum\tilde{\alpha}_i\D^i\1_j=0 
\end{displaymath}
for $j=1,\ldots,M$ and the modular forms $\tilde{\alpha}_i=\frac{w_i}{w_M}$.

Since there is no modular form of weight $2$, one needs $w_{M-1}=0$.
This can be checked explicitely as follows:
Let 
\begin{displaymath}
\1_{\kappa_s}
=q^{\kappa_s}(1+O(q))\:.
\end{displaymath}
We have
\begin{displaymath}
\D^m\1_{\kappa_s}
=\prod_{\ell=0}^{m-1}\left(q\frac{d}{dq}-\frac{\ell}{6}\right)q^{\kappa_s}(1+O(q))
\end{displaymath}
% since $E_{2k}=1+O(q)$ as well.
Thus for $M<12$,
\begin{displaymath}
q^{-\kappa_s}\mathcal{L}^{(2,\nu)}q^{\kappa_s}
=\sum_{m=0}^M\alpha_m\prod_{\ell=0}^{m-1}\left(\kappa_s-\frac{\ell}{6}\right)
\:.
\end{displaymath}
On the other hand,
since $\mathcal{L}^{(2,\nu)}\1_{\kappa_s}=0$ for $s=1,\ldots,M$, 
we have
\begin{displaymath}
\sum_{m=0}^M\alpha_m\prod_{\ell=0}^{m-1}\left(\kappa_s-\frac{\ell}{6}\right)
=\prod_{s=1}^{M}(\kappa-\kappa_s)
\end{displaymath}
The coefficient of $\kappa^{M-1}$ yields 
\begin{displaymath}
\alpha_{M-1}
=\sum_{s=1}^{M}\kappa_s
+\sum_{\ell=1}^{M}\frac{1-\ell}{6}\:,
\end{displaymath}
which vanishes.
\end{proof}

\subsection{Variation of the conformal structure}\label{Subsection: variation of the conformal structure for g=1}

Throughout this section, $\Sigma: y^2=p$ is the genus $1$ Riemann surface defined by
\begin{displaymath}%\label{eq: product rep of p with leading coefficient 4} 
p(x)
=4(x-X_1)(x-X_2)(x-X_3)\:.
\end{displaymath}
We assume that
\begin{align}\label{eq: vanishing condition on X}
\sum_{i=1}^3X_i=0
\:, 
\end{align}
or equivalently, there exist $a,b\in\C$ such that
\begin{align}\label{def: polynomial p defining the torus as doubl cover of CP1}
p(x)
=4(x^3+ax+b)
\:. 
\end{align}
We shall use the following notation:
Let $m(X_1,\xi_1,\ldots,X_n,\xi_n)$ be a monomial.
We denote by 
\begin{displaymath}
\overline{m(X_1,\xi_1,\ldots,X_n,\xi_n)} 
\end{displaymath}
the sum over all distinct monomials 
$m(X_{\sigma(1)},\xi_{\sigma(1)},\ldots,X_{\sigma(n)},\xi_{\sigma(n)})$,
where $\sigma$ is a permutation of $\{1,\ldots,n\}$.
E.g.\ for $n=3$, eq.\ (\ref{eq: vanishing condition on X}) reads $\overline{X}_1=0$, 
and in eq.\ (\ref{def: polynomial p defining the torus as doubl cover of CP1}), 
\begin{displaymath}
a
=\overline{X_1X_2}
=\underset{i<j}{\sum_{i,j=1}^3}X_iX_j\:,
=X_1X_2+X_1X_3+X_2X_3\:,
\quad
b
=-X_1X_2X_3\:.
\end{displaymath}
The normalised discriminant is \cite[p.\ 87]{S:1973}
\begin{align}\label{Delta(0) in terms of a,b}
\Delta
:=-4a^3-27b^2
\:.
\end{align}
Let 
\begin{displaymath}
V_3
:=\begin{pmatrix}
1&X_1&X_1^2\\
1&X_2&X_2^2\\
1&X_3&X_3^2
\end{pmatrix}
\end{displaymath}
be the $3\times 3$ Vandermonde matrix.
We have
\begin{displaymath}
\det V_3
=\prod_{1\leq i<j\leq 3}(X_j-X_i)
=(X_1-X_2)(X_2-X_3)(X_3-X_1) 
\end{displaymath}
so the discriminant equals
\begin{displaymath}
\Delta
=(\det V_3)^2\:.%\label{discriminant of p/4} 
\end{displaymath}
It will be convenient to work with the $1$-form 
\begin{displaymath}%\label{def: omega}
\omega
:=d\log\det V_3
%=\frac{1}{2}d\log\Delta
=\frac{\xi_1-\xi_2}{X_1-X_2}+\cyc
\:,
\end{displaymath}
where $\xi_j=dX_j$ ($j=1,2,3$), and for $k\geq 0$, with the matrix
\begin{displaymath}
\Xi_{3,k}
:=\begin{pmatrix}
X_1&X_2&X_3\\
1&1&1\\
\xi_1X_1^k&\xi_2X_2^k&\xi_3X_3^k
\end{pmatrix}\:.
\end{displaymath}
We have
\begin{align}\label{eq: quotient of determinants for Xi(3,k)}
\frac{\det\Xi_{3,k}}{\det V_3} 
\=\frac{\xi_1X_1^k}{(X_1-X_2)(X_3-X_1)}+\cyc
=-4\sum_{s=1}^3\frac{\xi_sX_s^k}{p'(X_s)}\:.
\end{align}

\begin{proposition}
We have
\begin{align}
\det\Xi_{3,1}
\=-\frac{\omega}{3}\det V_3
\:,\label{eq: omega as difference of xi's over difference of Xi's}\\
\det\Xi_{3,2} 
\=-\frac{a}{3}\det\Xi_{3,0}
\:.\label{eq: quotient of determinants for Xi(3,2)}
\end{align}
\end{proposition}

\begin{proof}
By eq.\ (\ref{eq: vanishing condition on X}),
\begin{displaymath}
d\det V_3
=-3X_1(d X_1)(X_2-X_3)+\cyc
=-3\det\Xi_{3,1}\:,
\end{displaymath} 
and eq.\ (\ref{eq: omega as difference of xi's over difference of Xi's}) follows.
We observe that
\begin{align}
\frac{1}{(X_1-X_2)(X_3-X_1)}+\cyc
\=0\label{cyclic sum of 1 over Ni is zero}\:,\\
\frac{X_1}{(X_1-X_2)(X_3-X_1)}+\cyc
\=0\label{cyclic sum of Xi over Ni is zero}
\:.  
\end{align}
By eq.\ (\ref{eq: vanishing condition on X}),
we have
\begin{align}\label{X1 squared}
X_1^2
\=-X_1(X_2+X_3)
=-a+X_2X_3
\:.
\end{align}
Since $\xi_1X_2X_3+\cyc=\overline{\xi_1X_2X_3}$, 
we have by eq.\ (\ref{cyclic sum of 1 over Ni is zero}),
\begin{align*}
\frac{\xi_1X_2X_3}{(X_1-X_2)(X_3-X_1)}+\cyc
\=-\:\frac{\xi_2X_3X_1+\xi_3X_1X_2}{(X_1-X_2)(X_3-X_1)}+\cyc\:.
\end{align*}
Moreover, $\overline{\xi_1}=0$, so
\begin{align*}
-\:\frac{\xi_2X_3X_1+\xi_3X_1X_2}{(X_1-X_2)(X_3-X_1)}+\cyc
\=\left(\frac{\xi_1(X_3X_1+X_1X_2)}{(X_1-X_2)(X_3-X_1)}+\cyc\right)
+\left(\frac{(\xi_3X_3+\xi_2X_2)X_1}{(X_1-X_2)(X_3-X_1)}+\cyc\right)\\
\=a\frac{\det\Xi_{3,0}}{\det V_3}
-\left(\frac{\xi_1X_2X_3}{(X_1-X_2)(X_3-X_1)}+\cyc\right)\\
\-\left(\frac{\xi_1X_1^2}{(X_1-X_2)(X_3-X_1)}+\cyc\right)\:,
\end{align*}
using symmetry of $\overline{\xi_1X_1}$ and eq.\ (\ref{cyclic sum of Xi over Ni is zero}). 
We conclude that
\begin{align}\label{eq for xi1X2X3 over N1}
\frac{\xi_1X_2X_3}{(X_1-X_2)(X_3-X_1)}+\cyc
\=\frac{2a}{3}\frac{\det\Xi_{3,0}}{\det V_3}
\:.
\end{align}
\end{proof}

\begin{proposition}\label{propos: determinants in terms of a,b and da, db}
We have
\begin{align}
\det(\Xi_{3,0}V_3)
\=9b\:da-6 a\:db\:,\label{det Xi(3,0)V in terms of da and db}\\
\det(\Xi_{3,1}V_3)
\=2a^2\:da+9b\:db
\:.\label{det Xi(3,1)V in terms of da and db}
\end{align}
\end{proposition}

\begin{proof}
For either equation, we show that under specific additional assumptions, 
both sides are proportional to $\Delta$, with the same proportionality factor.
The proof of general statement is moved to the Appendix. 

We first address eq.\ (\ref{det Xi(3,1)V in terms of da and db}) under the assumption that $\xi\propto X$.
In this case we have on the l.h.s.,
\begin{displaymath}
\det\Xi_{3,1}|_{\xi\propto X}\det V_3
\propto 
-\det V_3^2
=-\Delta\:.%\label{det Xi(3,1)V equals -Delta(0) in case xi=X}
\end{displaymath}
On the other hand,
\begin{align*}
da
\=\overline{\xi_1X_2}
\propto 2\overline{X_1X_2}
=2a\:,\\
db
\=-\overline{\xi_1X_2X_3}
\propto -3\overline{X_1X_2X_3}=3b\:, 
\end{align*}
so in this case the r.h.s.\ of eq.\ (\ref{det Xi(3,1)V in terms of da and db}) equals $-\Delta$.
The general case is treated in Appendix \ref{Appendix: Proof of eq. det Xi(3,1)V in terms of da and db}.
In order to prove eq.\ (\ref{det Xi(3,0)V in terms of da and db}),
suppose first that
\begin{align}\label{assumption: xi propto (Xi squared minus xi 0)}
\xi_i\propto X_i^2-\xi_0\:,
\quad
\xi_0
:=\frac{1}{3}\left(\sum_{i=1}^3X_i^2\right) 
=\frac{1}{3}\;\overline{X_1^2}\:.
\end{align}
Note that this does not affect condition (\ref{eq: vanishing condition on X}).
On the l.h.s.\ of eq.\ (\ref{det Xi(3,0)V in terms of da and db}), 
we have by assumption (\ref{assumption: xi propto (Xi squared minus xi 0)}),
\begin{displaymath}
\det\Xi_{3,0}
=\det
\begin{pmatrix}
\xi_1&\xi_2&\xi_3\\
X_1&X_2&X_3\\
1&1&1
\end{pmatrix} 
\propto\: 
\det
\begin{pmatrix}
X_1^2&X_2^2&X_3^2\\     
X_1&X_2&X_3\\
1&1&1
\end{pmatrix}
-\det 
\begin{pmatrix}
\xi_0&\xi_0&\xi_0\\     
X_1&X_2&X_3\\
1&1&1
\end{pmatrix}\:,
\end{displaymath}
where the latter determinant is zero. 
Thus $\det\Xi_{3,0}\propto\det V_3$, and
\begin{align*}
\det\Xi_{3,0}|_{\xi\propto X^2-\xi_0}\det V_3
\propto\:&
\det
\begin{pmatrix}
\overline{\xi_1}&\overline{\xi_1X_1}&\overline{\xi_1X_1^2}\\
\overline{X_1}&\overline{X_1^2}&\overline{X_1^3}\\
3&\overline{X_1}&\overline{X_1^2}\\
\end{pmatrix}
=-\Delta\:.
\end{align*}
On the other hand, by the fact that $\overline{X_1}=0$, we have
\begin{align*}
\xi_0
\=\frac{1}{3}\:\overline{X_1^2}
=-\frac{2}{3}\:\overline{X_1X_2} 
=-\frac{2a}{3}\:,\\
\overline{X_1^3}
\=-3\overline{X_1^2X_2}-6b\:,\\
\overline{X_1^2X_2}
\=\overline{X_1X_2(X_1+X_2)}
=-3b\:.
\end{align*}
So on the r.h.s.\ of eq.\ (\ref{det Xi(3,0)V in terms of da and db}), 
\begin{align*}
d a
\=-\overline{\xi_1X_1} 
\propto-\overline{X_1^3}+\xi_0\overline{X_1}
=-\overline{X_1^3}
=-3b\:,\\
d b
\=-\overline{\xi_1X_2X_3}
\propto-\overline{X_1^2X_2X_3}+\xi_0\overline{X_1X_2}
=b\overline{X_1}+\xi_0 a
=\xi_0 a
=-\frac{2}{3}\:a^2\:.
\end{align*}
From this and eq.\ (\ref{Delta(0) in terms of a,b}) follows eq.\ (\ref{det Xi(3,0)V in terms of da and db}).
The general case without the assumption (\ref{assumption: xi propto (Xi squared minus xi 0)}) 
is proved in Appendix \ref{Appendix: Proof of eq. det Xi(3,0)V in terms of da and db}.

\end{proof}

\begin{lemma}\label{Lemma: formula for the connection 1-form omega in terms of tau and lambda}
% Let $\Sigma:y^2=p$ be the Riemann surface defined by eq.\ (\ref{def: polynomial p defining the torus as doubl cover of CP1}).
% Define a deformation of $\Sigma$ by
% \begin{displaymath}
% \xi_j=d X_j\:,\quad j=1,2,3. 
% \end{displaymath}
In terms of the modulus $\tau$ and the scaling parameter $\ell$ (the inverse length of the real period), 
we have
\begin{displaymath}
\omega
%\=\pi\i E_2\:d\tau-6\frac{d\ell}{\ell}
=\pi\i E_2\:d\tau-6d\log\ell\:.
%\label{omega terms of a and b}
\end{displaymath}
\end{lemma}

\begin{proof}
In eq.\ (\ref{def: polynomial p defining the torus as doubl cover of CP1}),
we have \cite{S:1973}
\begin{align}\label{a,b in terms of the Eisenstein series}
a
=-\frac{\pi^4}{3}\ell^4E_4\:,
\hspace{1cm}
b
=-\frac{2\pi^6}{27}\ell^6E_6\:, 
\end{align}
so by eq.\ (\ref{Delta(0) in terms of a,b}),
\begin{align}\label{Delta(0) in terms of tau and lambda}
\Delta
=\frac{4\pi^{12}}{27}\ell^{12}(E_4^3-E_6^2)
\:.
\end{align}
Using that
\begin{align}\label{Serre derivative of E4 and E6}
\D E_4=-\frac{E_6}{3}\:,
\quad\quad
\D E_6=-\frac{E_4^2}{2}\:
\end{align}
\cite[Proposition 15, p.\ 49]{Z:1-2-3}, where $\D$ is the Serre derivative,
we find 
%by eq.\ (\ref{det Xi(3,1)V in terms of da and db}),
\begin{displaymath}
2\:a^2\frac{\partial}{\partial\tau}a+9\:b\:\frac{\partial}{\partial\tau}b
=-\frac{\i\pi}{3}E_2\Delta
\:.                                             
\end{displaymath}
% so 
% \begin{displaymath}
% \frac{\partial}{\partial\tau}\Delta
% =-6\frac{\partial}{\partial\tau}\left(\frac{2}{3}a^3+\frac{9}{2}b^2\right)
% =2\pi iE_2\Delta
% \end{displaymath}
From eqs (\ref{a,b in terms of the Eisenstein series}) and (\ref{Delta(0) in terms of tau and lambda}) follows
\begin{displaymath}
2a^2\frac{\partial}{\partial\ell}a+9b\frac{\partial}{\partial\ell}b
=-\frac{2}{\ell}\Delta
\:.
\end{displaymath}
So by eq.\ (\ref{det Xi(3,1)V in terms of da and db}),
\begin{displaymath}
\frac{\det\Xi_{3,1}}{\det V_3}
=\frac{\i\pi}{3}E_2d\tau-2d\log\ell
\end{displaymath}
and the proposition follows.
\end{proof}

Under variation of the ramification points, the modulus changes according to

\begin{lemma}\label{Lemma: d tau in terms of the quotient of determinants}
Under the conditions of Lemma \ref{Lemma: formula for the connection 1-form omega in terms of tau and lambda},
we have
\begin{align}\label{eq: identity for d tau} 
d\tau
\=-\pi i\ell^2\:\frac{\det\Xi_{3,0}}{\det V_3}\:.
\end{align}
\end{lemma}

Note that proportionality between the differentials on either side of eq.\ (\ref{eq: identity for d tau}) can be seen as follows:
Under the action of 
$\begin{pmatrix}
a&b\\
c&d
\end{pmatrix}\in SL_2(\Z)$, 
both $d\tau$ and $\ell^2$ (the squared inverse length of the real period) transform by a factor of $(c\tau+d)^{-2}$.  
Moreover, both differentials have a simple pole at the boundary of the moduli space: 
$d\tau$ is singular at $\tau=i\infty$, 
while $\frac{\det\Xi_{3,0}}{\det V_3}$ has a pole when two $X_i$ coincide. 
Thus up to a multiplicative constant they must be equal.

\begin{proof}[Proof of Lemma \ref{Lemma: d tau in terms of the quotient of determinants}]
By eqs (\ref{a,b in terms of the Eisenstein series}) -- (\ref{Serre derivative of E4 and E6}),
\begin{displaymath}
9\:b\:\frac{\partial}{\partial\tau}a-6\:a\:\frac{\partial}{\partial\tau}b
=2\pi\i(9\:b\:\D a-6\:a\:\D b)\\
=\frac{\i}{\pi\ell^2}\Delta
\:.                                             
\end{displaymath}
The partial derivatives are actually ordinary derivatives since by eqs (\ref{a,b in terms of the Eisenstein series}),
\begin{align*}
9\:b\:\frac{\partial}{\partial\ell}a-6\:a\:\frac{\partial}{\partial\ell}b
\=0\:.
\end{align*}
Factoring out $d\tau$ on the r.h.s.\ of eq.\ (\ref{det Xi(3,0)V in terms of da and db}) and dividing both sides by $\Delta/(-i\pi\ell^2)$ 
yields the propositioned formula. 
\end{proof}

\subsection{Explicit results for $g=1$}

This section largely uses results obtained for arbitrary genus in \cite{L:PhD14,LN:2017}
though Theorem \ref{Theorem: system of differential equations in algebraic formulation} 
is proved independently using the methods introduced in Subsection \ref{Subsection: variation of the conformal structure for g=1}.
It is shown (Proposition \ref{proposition: equivalence of two systems of ODE for g=1}) 
that the two formulations are equivalent for $g=1$.

\subsubsection{General results}\label{subsection: general results}

Let $1$ and $T$ be the identity field and the Virasoro field, respectively, on $\Sigma$.
It will be useful to work with the field \cite{L:2013,L:PhD14}
\begin{align}\label{def: vartheta}
\vartheta(x)
:=T(x)\:p(x)-\frac{c}{32}\frac{[p'(x)]^2}{p(x)}.1\:.
\end{align}
where $c\in\R$ is the central charge.

In \cite{LN:2017}, we defined a singular metric on $\Sigma$ which is obtained by lifting a polyhedral metric on $\CP^1$,
whose curvature is concentrated on the set of ramification points and equally distributed over this set. 
Let $\1$ be the $0$-point function corresponding to the singular metric.
In our algebraic approach, a second parameter is given by the unknown constant $\A$, 
which is proportional to $\1$ but with an unknown proportionality factor:
\begin{displaymath}
\A
=:\alpha\1 
\:.
\end{displaymath}
We have
\begin{displaymath}
\langle\vartheta(x)\rangle
=-cx\1+\frac{\A}{4} 
\:.
\end{displaymath}
It will be useful to work with the auxiliary function
\begin{displaymath}%\label{def: Gamma}
\Gamma(x)
:=cx\1-\frac{\A}{2}
\:.
\end{displaymath}
Now Theorem 2 in \cite{L:PhD14, L:2013} yields:

\begin{proposition}\label{proposition: <TT> connected}
Let $\Sigma:y^2=p$ be the Riemann surface defined by eq.\ (\ref{def: polynomial p defining the torus as doubl cover of CP1}).
For $|x_1|,|x_2|$ small, we have
\begin{equation}\label{P[1] for g=1}
\begin{split}
\left\{\langle T(x_1)T(x_2)\rangle
-\1^{-1}\langle T(x_1)\rangle\langle T(x_2)\rangle\right\}&p(x_1)p(x_2)\\
% \=R(x_1,x_2)
% -\frac{1}{2}\left(x_1\Theta(x_2)+x_2\Theta(x_1)\right)
% -2cx_1x_2\1
% +P^{[1]}\\
=R(x_1,x_2)
+x_1\Gamma(x_2)\+x_2\Gamma(x_1)
+P^{[1]}
\:.
\end{split}
\end{equation}
Here $P^{[1]}$ is constant in position, and
\begin{displaymath}
R(x_1,x_2)
=R^{[1]}(x_1,x_2)
+y_1y_2R^{[y_1y_2]}(x_1,x_2)
\:,
\end{displaymath}
where $R^{[y_1y_2]}(x_1,x_2)$ is a rational function of $x_1$ and $x_2$ and
\begin{equation*}%\label{def: R[1] von <TT>c}
\begin{split}
R^{[1]}&(x_1,x_2)\\
\=\frac{c}{4}\frac{p(x_1)p(x_2)}{(x_1-x_2)^4}\1
+\frac{c}{32}\frac{p'(x_1)p'(x_2)}{(x_1-x_2)^2}\1
+\frac{1}{2}\frac{p(x_2)\langle\vartheta(x_1)\rangle+p(x_1)\langle\vartheta(x_2)\rangle}{(x_1-x_2)^2}
\:.
\end{split}
\end{equation*}
% Moreover, let $p^{[1]}$ and $p^{[x]}$ be even polynomials such that $p(x)=p^{[1]}(x)+xp^{[x]}(x)$.
% We have
% \begin{equation*}
% \begin{split}
% R^{[y_1y_2]}(x_1,x_2)
% \=\frac{c}{4}\1\left(\frac{p^{[1]}(\sqrt{x_1x_2})}{(x_1-x_2)^4} 
% +\frac{1}{2}(x_1+x_2)\frac{p^{[x]}(\sqrt{x_1x_2})}{(x_1-x_2)^4}\right)
% +\frac{\A}{4}\frac{1}{(x_1-x_2)^2}
% \\
% \=\frac{c}{2}\frac{(x_1+x_2)(x_1x_2+a)+2b}{(x_1-x_2)^4}\1
% +\frac{\A}{4}\frac{1}{(x_1-x_2)^2}
%\:.
% \end{split}
% \end{equation*}
\end{proposition}

For $P^{[1]}$ defined by eq.\ (\ref{P[1] for g=1}), 
let
\begin{displaymath}%\label{eq: P in relation to P[1]}
P
:=P^{[1]}+\frac{1}{16}\alpha\A
\:.
\end{displaymath}
In terms of the one-point function of the normal ordered product of $T$,
% \begin{align}\label{def: Lambda}
% \Lambda(x)
% =\langle N_0(T,T)(x)\rangle-\frac{3}{10}\partial^2T(x)
% \:,
% \end{align}
\begin{displaymath}
\langle N_0(T,T)(x_2)\rangle
=\lim_{x_1\rechts x_2}\left[
\langle T(x_1)T(x_2)\rangle
-\frac{c/2}{(x_1-x_2)^4}\1
-\frac{1}{(x_1-x_2)^2}\langle T(x_1)+T(x_2)\rangle\right] 
\:,
\end{displaymath}
we have by the theorem, by eq.\ (\ref{def: vartheta}) and by the fact that $\1^{-1}\langle\vartheta(x)\rangle^2=cx\Gamma(x)+\frac{\alpha\A}{16}$,
\begin{align*}
P
\=p^2\langle N_0(T,T)\rangle
-\left(\frac{c}{32}\right)^2\frac{[p']^4}{p^2}-\frac{c}{16}\frac{[p']^2}{p}\langle\vartheta\rangle
-\oint_{x_2}\frac{R(x_1,x_2)}{x_1-x_2}\frac{dx_1}{2\pi i}
-(c+2)x\Gamma(x)
\:.
\end{align*}

\begin{proposition}\label{proposition: 2-point function of vartheta}
Let $\Sigma:y^2=p$ be the Riemann surface defined by eq.\ (\ref{def: polynomial p defining the torus as doubl cover of CP1}).
%For $|x_1|,|x_2|$ small, we have
We have for $x_1$ close to $x_2$, 
\begin{equation*}%\label{eq.: graphical representation of vartheta's}
\langle\vartheta(x_1)\vartheta(x_2)\rangle
=R_*(x_1,x_2)
+\langle\vartheta(x_1)\vartheta(x_2)\rangle_r
\:. 
\end{equation*}
Here $\langle\vartheta(x_1)\vartheta(x_2)\rangle_r$ is regular at $x_1=x_2$.
We have
\begin{enumerate}
 \item
 %\begin{remark}
%We have %\cite{L:PhD14,L:2013}
\begin{displaymath}
R_*(x_1,x_2)
=R_*^{[1]}(x_1,x_2)+y_1y_2R_*^{[y_1y_2]}(x_1,x_2)
\end{displaymath}
where $R_*^{[y_1y_2]}$ is a rational function in $x_1$ and $x_2$, and
\begin{align*}
R_*^{[1]}(x_1,x_2)
% \=\frac{c}{4}\frac{p(x_1)p(x_2)}{(x_1-x_2)^4}\1
% +\frac{c}{32}\frac{(p(x_1)-p(x_2))^2}{(x_1-x_2)^4}\1
% +\frac{1}{4}\frac{p(x_1)+p(x_2)}{(x_1-x_2)^2}\left\{\langle\vartheta(x_1)\rangle+\langle\vartheta(x_2)\rangle\right\}\\
\=R^{[1]}(x_1,x_2)
-\left\{\frac{5c}{4}(c+2)x_1x_2
+\frac{3c}{2}a
+\frac{5c}{4}(x_1-x_2)^2\right\}\1
\:.
\end{align*}
% The constant term in $R_*^{[1]}(x_1,x_2)$ are proportional to $\A$ 
% so that the term $\propto\frac{3c}{2}a\1$ is compensated for by the corresponding term in $R^{[1]}(x_1,x_2)$.
%\end{remark}
\item\label{proposition: <vartheta vartheta>r}
\begin{displaymath}
\langle\vartheta(x_1)\vartheta(x_2)\rangle_r
=\langle\vartheta(x_1)\vartheta(x_2)\rangle_r^{[1]}+y_1y_2\langle\vartheta(x_1)\vartheta(x_2)\rangle_r^{[y_1y_2]}
\end{displaymath}
where $\langle\vartheta(x_1)\vartheta(x_2)\rangle_r^{[y_1y_2]}$ is a polynomial in $x_1$ and $x_2$, and
\begin{equation}\label{eq: <vartheta 1 vartheta 2>r with the constant not yet determined} 
%\begin{split}
\langle\vartheta(x_1)\vartheta(x_2)\rangle_r^{[1]}
=\frac{1}{2}(c+2)\left\{x_1\Gamma(x_2)+x_2\Gamma(x_1)\right\}
+D
+\frac{c}{2}(x_1-x_2)^2\1
% \\
% %
% \=x_1\Gamma(x_2)+x_2\Gamma(x_1)+P^{[1]}\\
% \+\frac{c}{2}\left\{x_1\Gamma(x_2)+x_2\Gamma(x_1)\right\}
% +(D-P^{[1]})
% +\frac{c}{2}(x_1-x_2)^2\1
\:.
\nn
%\end{split}
\end{equation}
Here $D$ is constant in position. 
\end{enumerate}
\end{proposition}

\begin{proof}
 \begin{enumerate}
  \item By the graphical representation theorem \cite{LN:2017}, 
\begin{equation}\label{eq.: graphical representation of vartheta's}
\begin{split}
\langle\vartheta(x_1)\vartheta(x_2)\rangle
\=R_*(x_1,x_2)
+\langle\vartheta(x_1)\vartheta(x_2)\rangle_r
\:,\\
R_*(x_1,x_2)
:\=\frac{c}{32}f_{12}^2\1+\frac{1}{4}f_{12}\langle\vartheta(x_1)+\vartheta(x_2)\rangle
\:, 
\end{split}
\end{equation} 
where $f_{12}:=(y_1+y_2)^2/(x_1-x_2)^2$.
$R_*^{[1]}(x_1,x_2)+$ and $R_*^{[y_1y_2]}(x_1,x_2)$ are obtained by direct computation.
% \begin{proof}[Proof of the Remark]
% The formula for $R_*^{[1]}$ is obtained by direct computation, see also \cite{L:PhD14,L:2013}.
% The difference from $R^{[1]}$ follows from eq.\ (\ref{def: polynomial p defining the torus as doubl cover of CP1}) 
% and the identities
% \begin{align*}
% \frac{(p_1-p_2)^2}{(x_1-x_2)^4}
% =\frac{p_1'p_2'}{(x_1-x_2)^2}
% \+\frac{1}{2}\frac{p_1'p_2''-p_1''p_2'}{x_1-x_2}
% +4\:\frac{p_1-p_2}{x_1-x_2}\\
% \+2(p_1'+p_2')
% -\frac{1}{4}p_1''p_2''
% -(p_1''-p_2'')(x_1-x_2)
% \:.
% \end{align*}
% and for $\Theta(x)=4\langle\vartheta(x)\rangle$ and for $\Theta_i=\Theta(x_i)$,
% \begin{displaymath}
% \frac{p_1+p_2}{(x_1-x_2)^2}(\Theta_1+\Theta_2)
% =2\:\frac{p_1\Theta_2+p_2\Theta_1}{(x_1-x_2)^2}
% -4c\1\frac{p_1-p_2}{x_1-x_2}
% \:.
% \end{displaymath}
% \end{proof} 
\item For $x\rechts\infty$,
\begin{displaymath}
\vartheta(x)
=-cx.1+O(1)\:,
\end{displaymath}
so for large $x$, and for $s=1,2,3$, we have
\begin{align}\label{incomplete eq. for <vartheta(x)vartheta(Xs)> for n=3}
\langle\vartheta(x)\vartheta(X_s)\rangle
\=-cx\:\langle\vartheta(X_s)\rangle+O(1)
\:,
\end{align}
On the other hand, for the l.h.s.\ we have by the graphical representation eq.\ (\ref{eq.: graphical representation of vartheta's}),
we have for $x\gg X_s\gg 1$,
\begin{align*}
f_{xX_s}
\=\frac{p(x)}{(x-X_s)^2} 
%\=4(x^3+ax+b)\:x^{-2}\left(1+2\frac{X_s}{x}+3\left(\frac{X_s}{x}\right)^2+O(x^{-3})\right)\\
=4\left(x+2X_s+(a+3X_s)x^{-1}\right)+O(x^{-2})\\
f_{xX_s}^2
\=16(x^2+4xX_s+2a+10X_2^2)+O(x^{-1})
\:.
\end{align*}
% (Since the $O(1)$ terms are unknown on the r.h.s.\ of eq.\ (\ref{incomplete eq. for <vartheta(x)vartheta(Xs)> for n=3}),
% we only determine $f_{xX_s}$ and $f_{xX_s}^2$ on the l.h.s.\ up to (and including) order $O(1)$ and $O(x)$, respectively.)
%Actually the $O(x^{-1})$ term in $f_{xX_s}$ and the $O(1)$ term in $f_{xX_s}^2$ contributions to eq.\ (\ref{eq.: graphical representation of vartheta's}) will drop out.
So in the region considered,
\begin{align*}
R_*(x,X_s)
\=\frac{c}{32}f_{xX_s}^2\1
+\frac{1}{4}f_{xX_s}\langle\vartheta(x)+\vartheta(X_s)\rangle\\
\=\left(\frac{1}{2}x+X_s\right)\left(-cx\1+\A\right)
+O(x^{-1})
\end{align*}
$\langle\vartheta(x)\vartheta(X_s)\rangle_r$ is a polynomial in $x-X_s$, thus regular at $x=0$.
Thus by eqs\ (\ref{eq.: graphical representation of vartheta's}) and (\ref{incomplete eq. for <vartheta(x)vartheta(Xs)> for n=3}),
\begin{displaymath}
\langle\vartheta(x)\vartheta(X_s)\rangle_r
=-cx\langle\vartheta(X_s)\rangle
+cx\left(\frac{1}{2}x+X_s\right)\1
-\frac{1}{2}x\A
+D_s
\end{displaymath}
%We have put the term $-X_s\A$ into the definition of $D_s$, 
%accounting for $-\frac{1}{4}(c+2)X_s\A$ in place of $-\frac{1}{4}(c-2)X_s\A$.
where by symmetry, 
\begin{displaymath}
D_s
=\frac{c}{2}X_s^2\1 
-\frac{1}{4}(c+2)X_s\A
+D\:,
\end{displaymath}
with $D\in\1\C$. 
This yields eq.\ (\ref{eq: <vartheta 1 vartheta 2>r with the constant not yet determined}) for arbitrary $x_1,x_2$.
(Note that the degree $2$ polynomial in $x_2$,
\begin{align*}
c_1x_1^2x_2^2+c_2(x_1^2x_2+x_1x_2^2)\:,
\quad
c_1,c_2\in\C\:,
\end{align*}
has a zero at $X_s$ for $s=1,2,3$ and so is identically zero: $c_1=c_2=0$.)
\end{enumerate}
\end{proof}

Our formulae single out the case $c=-2$, whose significance is unclear to us at this stage.
Since it corresponds to a non-minimal model we won't address it here.

\begin{theorem}\label{Theorem: system of differential equations in algebraic formulation}
Let $\Sigma:y^2=p$ be the Riemann surface defined by eq.\ (\ref{def: polynomial p defining the torus as doubl cover of CP1})
and equipped with the singular metric \cite{LN:2017}.
Define a deformation of $\Sigma$ by
\begin{displaymath}
\xi_j=d X_j\:,\quad j=1,2,3\:. 
\end{displaymath}
We have the following system of linear differential equations
\begin{equation}\label{system: variation of zero-point function and A1 in terms of quotients of determinants}
\begin{split}
\left(d+\frac{c}{24}\:\omega\right)\:\1 
\=-\frac{1}{8}\A \:\frac{\det\Xi_{3,0}}{\det V_3}\:,\\
\left(d+\frac{c}{24}\omega\right)\:\A 
\=\:C\:\frac{\det\Xi_{3,0}}{\det V_3}
-\A \:\frac{\det\Xi_{3,1}}{\det V_3}\:,
\end{split}
\end{equation}
where $\omega$ is be the $1$-form defined by eq.\ (\ref{eq: omega as difference of xi's over difference of Xi's}), 
and 
\begin{displaymath}
C:
=-2P-\frac{8c}{3}a\1\:.
\end{displaymath}
\end{theorem}

The formulation of the differential equations using determinants relies on the permutation symmetry of the equations' constituent parts.
This symmetry will continue to be present as the number of ramification points increases.
With the genus, however, also the degree of the polynomial $\langle\vartheta\rangle$ will grow 
and give rise to additional terms having no lower genus counterpart.

\begin{proof}[Proof of the Theorem]
%Notations: All state-dependent objects are understood to refer to the singular metric on $\Sigma$.
To simplify notations, set $\Theta(x):=4\langle\vartheta(x)\rangle$.
The following two identities will be useful:
\begin{align}
\frac{dp}{p}
\=-\sum_{s=1}^3\frac{\xi_s}{x-X_s}\:,\label{dp/p}\\ 
d\left(\frac{p'}{p}\right)
\=\sum_{s=1}^3\frac{\xi_s}{(x-X_s)^2}\:.\label{d(p'/p)}
%=\sum_{s=1}^n\frac{\xi_s}{x^2}\left(\sum_{k=0}^{\infty}\left(\frac{X_s}{x}\right)^k\right)^2\:\:.
\end{align}
Let $\gamma_1$ be a closed path enclosing $X_1\in\CP^1$ and no other zero of $p$. 
$x$ does not define a coordinate close to $X_1$, however $y$ does. 
On the ramified covering, a closed path winds around $X_1$ by an angle of $4\pi$. 
We will work with the $x$ coordinate, 
and mark the double circulation along $\gamma_1$ in $\CP^1$ by a symbolic $2\times\gamma_1$ under the integral. 
When the integrand is Galois-even, the double winding results in a factor of two.
On the other hand, the contour integral of a Galois-odd term corresponding to an angle of $4\pi$ vanishes
when the square root of $p(x)$ changes sign after the first full $2\pi$ circle.
Using eqs (\ref{eq: vanishing condition on X}) on the ramification points, 
and (\ref{def: vartheta}) and (\ref{dp/p}), 
we obtain
\begin{align*}
\frac{d}{dX_1}\1
=\frac{1}{2\pi\i}\ointctrclockwise_{2\times\gamma_1}\langle T(x)\rangle\:&\:dx
=2\lim_{x\rechts X_1}(x-X_1)
\langle T(x)\rangle
\nn\\
\=\frac{1}{8}\left(\frac{c\1}{X_1-X_2}
+\frac{c\1}{X_1-X_3}
+\frac{\Theta(X_1)}{(X_1-X_2)(X_1-X_3)}\right)\\
\=-\frac{1}{2}\frac{(2X_1+X_2+X_3)c\1-\A }{p'(X_1)}\\
\=-\frac{1}{2}\frac{X_1c\1}{p'(X_1)}+\frac{1}{2}\frac{\A }{p'(X_1)}
\:.
\end{align*}
so
\begin{align}\label{eq: d<1>}
d\1
=\frac{c}{8}\frac{\det\Xi_{3,1}}{\det V_3}\1-\frac{1}{8}\frac{\det\Xi_{3,0}}{\det V_3}\A 
\:.
\end{align}
The same argument applies when $X_1$ is replaced with $X_2$ or $X_3$. 
Using eq.\ (\ref{eq: quotient of determinants for Xi(3,k)}),
we recover the differential equation (\ref{system: variation of zero-point function and A1 in terms of quotients of determinants}) for $\1$.
When $\langle T(x)\rangle$ is varied by changing all ramifications points $X_1,X_2,X_3$ simultaneously, 
we must require the position $x$ not to lie on or be enclosed by any of the corresponding three curves 
$\gamma_1$, $\gamma_2$ and $\gamma_3$.
Then we have  
\begin{displaymath}
d\langle T(x)\rangle-\langle T(x)\rangle\:d\log\1
=\sum_{j=1}^3\left(\frac{\1}{2\pi\i}\ointctrclockwise_{2\times\gamma_j}\:\langle T(x')T(x)\rangle_c\:dx'\right)\:\xi_j
\:,
\end{displaymath}
where the connected $2$-point function is given by Proposition \ref{proposition: <TT> connected}.
As mentioned before, the term $y'yR^{[y'y]}(x',x)$ (with ${y'}^2=p(x')$) does not contribute.
For $j=1$ we have
\begin{align}
\ointctrclockwise_{2\times\gamma_1}\:\frac{R^{[1]}(x',x)}{p(x')p(x)}\:\frac{dx'}{2\pi\i}
% \=2\lim_{x'\rechts X_1}(x'-X_1)
% \frac{R^{[1]}(x',x)}{p(x')p(x)}\nn\\
\=\frac{c}{16}\frac{\1}{(X_1-x)^2}\:\frac{p'(x)}{p(x)}%\label{the first term corresponds to d of [p'] squared over p squared}\\
+\frac{1}{4}\frac{1}{(X_1-x)^2}\frac{\Theta(X_1)}{p'(X_1)}\nn
\end{align}
Moreover, 
\begin{align*}
\ointctrclockwise_{2\times\gamma_1}&
\:\frac{P^{[1]}-\frac{1}{2}\left(x'\Theta(x)+x\Theta(x')\right)-2cx'x\1}{p(x')p(x)}\:\frac{dx'}{2\pi\i}\\
% \=\lim_{x'\rechts X_1}(x'-X_1)
% \frac{2P^{[1]}-x\Theta(x')-x'\A}{p(x')p(x)}\\
\=
\frac{2P^{[1]}}{p(x)}\frac{1}{p'(X_1)}
-\frac{x}{p(x)}\frac{\Theta(X_1)}{p'(X_1)}
-\frac{\A }{p(x)}\frac{X_1}{p'(X_1)}
\:.
\end{align*}
Using eq.\ (\ref{d(p'/p)}), we obtain
\begin{align}
 d\langle T(x)\rangle
\-\langle T(x)\rangle\:d\log\1
-\frac{c}{32}\1d\left[\frac{p'}{p}\right]^2\nn\\
\=\left(\frac{1}{4}\frac{1}{(X_1-x)^2}-\frac{x}{p(x)}\right)\frac{\xi_1\Theta(X_1)}{p'(X_1)}+\cyc\label{Theta expression}\\
\-\frac{1}{4p(x)}\left(2P^{[1]}\frac{\det\Xi_{3,0}}{\det V_3}-\A \frac{\det\Xi_{3,1}}{\det V_3}\right)
\:.\nn
\end{align}
We deduce the differential equation for $\A $ using the identity
\begin{displaymath}
d\Theta
=4p\:d\left(\frac{\Theta}{4p}\right)+\Theta\:d\log p
\:. 
\end{displaymath}
By eq.\ (\ref{def: vartheta}),
\begin{align*}
4p\:d\left(\frac{\Theta}{4p}\right)|_x
=4p\left(d\langle T(x)\rangle
-\langle T(x)\rangle\:d\log\1
-\frac{c}{32}\1d\left[\frac{p'}{p}\right]^2\right)+\Theta\:d\log\1
\:.
\end{align*}
Now we address $\Theta\:d\log p$. Using partial fraction decomposition of $\Theta/p$, 
\begin{align*}
\frac{\Theta(x)}{p(x)}
=\frac{1}{(x-X_1)}\:\frac{\Theta(X_1)}{p'(X_1)} 
+\cyc.
\end{align*}
multiplying by $p$ and using eq.\ (\ref{dp/p}) yields
\begin{align}
(\Theta\:d\log p)|_x
\=-p(x)\left(\frac{1}{(x-X_1)}\frac{\Theta(X_1)}{p'(X_1)}+\cyc\right)\left(\frac{\xi_j}{(x-X_j)}+\cyc\right)\:.
\label{first term on the r.h.s. of Theta times (xi over Ni +cyc)}
\end{align}
Overall three terms on the r.h.s.\ of eq.\ (\ref{first term on the r.h.s. of Theta times (xi over Ni +cyc)}) 
are equal but opposite to $4p$ times the first term (and its cyclic permutations) in line (\ref{Theta expression}). 
Since $\overline{\xi_1}=0$, we have for the remaining terms in eq.\ (\ref{first term on the r.h.s. of Theta times (xi over Ni +cyc)}),
\begin{align}
-p(x)&\left(\frac{1}{(x-X_1)}\frac{\Theta(X_1)}{p'(X_1)}\left(\frac{\xi_2}{(x-X_2)}+\frac{\xi_3}{(x-X_3)}\right)+\cyc\right)\nn\\
\=4\left((\xi_2X_3+\xi_3X_2)\frac{\Theta(X_1)}{p'(X_1)}+\cyc\right)
+4x\left(\frac{\xi_1\Theta(X_1)}{p'(X_1)}+\cyc\right)\:,\label{line: terms that remain from Theta d log p}
\end{align}
The second term in line (\ref{line: terms that remain from Theta d log p})
is equal but opposite to $4p$ times the second term and its cyclic permutations in line (\ref{Theta expression}).
The first term in line (\ref{line: terms that remain from Theta d log p}) equals
(cf.\ Appendix \ref{Appendix Section: Completion of the proof of Theorem on DE in algebraic formulation})
\begin{align*}
4\left((\xi_2X_3+\xi_3X_2)\frac{\Theta(X_1)}{p'(X_1)}+\cyc\right)
\=-\frac{8c}{3}a\1\frac{\det\Xi_{3,0}}{\det V_3}
-2\A \frac{\det\Xi_{3,1}}{\det V_3}\:.
\end{align*}
Thus we have shown that
\begin{displaymath}%\label{eq: d Theta}
d\Theta
=-\left(2P^{[1]}
+\frac{8c}{3}a\1\right)\frac{\det\Xi_{3,0}}{\det V_3}
-\A \frac{\det\Xi_{3,1}}{\det V_3}
+\Theta\:d\log\1
\:.
\end{displaymath}
Taking eq.\ (\ref{eq: d<1>}) with eq.\ (\ref{eq: omega as difference of xi's over difference of Xi's}) into account,
\begin{align*}
d\A 
%\=d\Theta+4cx\:d\1\\
\=-\left(2P^{[1]}+\frac{8c}{3}\:a\1\right)\frac{\det\Xi_{3,0}}{\det V_3}
-\A \frac{\det\Xi_{3,1}}{\det V_3}
+\A \:d\log\1
\\
\=-\frac{c}{24}\omega\A 
-\left(2P+\frac{8c}{3}\:a\1\right)\frac{\det\Xi_{3,0}}{\det V_3}
-\A \frac{\det\Xi_{3,1}}{\det V_3}
\:.
\end{align*}
The coefficient of $\det\Xi_{3,0}/\det V_3$ defines $C$.
%, which for $c=-\frac{22}{5}$ equals $(22/75)a\1$ by eq.\ (\ref{eq: P in relation to P[1]}).
\end{proof}

\begin{remark}
Denote by $\oneflat$ and $\T$ be the $0$-point function and the position independent $1$-point function of the Virasoro field, respectively,
on the flat torus $(\Sigma,|dz|^2)$) with the analytic coordinate $z$.
Define $\Aflat:=4\ell^2\T$, where $\ell$ is the inverse length of the real period and 
\begin{displaymath}
\Aflat
=:\alphaflat\oneflat\:. 
\end{displaymath}
On the other hand, write $\onesing$ for the $0$-point function w.r.t.\ our singular metric on $\Sigma$, 
and 
\begin{displaymath}
\Asing
=:\alphasing\onesing 
\:.
\end{displaymath}
We have $\alphasing=\alphaflat$ \cite{LN:2017}.
Thus by the ODE (\ref{ODE for zero-point function in z coordinate in the (2,5) minimal model}) for $\oneflat$
and the ODE for $\onesing$ from system (\ref{system: variation of zero-point function and A1 in terms of quotients of determinants}),  
together with eq.\ (\ref{eq: identity for d tau}), 
\begin{displaymath}
d\log\frac{\onesing}{\oneflat}
=-\frac{c}{48}d\log\Delta
\:,
\end{displaymath}
so
\begin{align}\label{eq: identity relating 0-point function in flat metric to that in singular metric}
\onesing
=\Delta^{-\frac{c}{48}}\:\oneflat\:.
\end{align} 
In particular, $\onesing$ is not a modular function. 
\end{remark}

\subsubsection{Application to the $(2,5)$ minimal model}

In the $(2,5)$ minimal model, we have $c=-22/5$ and 
%in eq.\ (\ref{def: Lambda}), $\Lambda=0$.
\begin{displaymath}
N_0(T,T)
=\frac{3}{10}\partial^2T
\:. 
\end{displaymath}
The undeterminates of Subsection \ref{subsection: general results} take the values
\begin{equation*}
P
=\frac{143}{25}a\1\:,
\quad
C
=\frac{22}{75}\:a\1\:,
\quad
D
=-\frac{22}{25}a\1
\:,
\end{equation*}
respectively.
Indeed, by \cite[Lemma 16]{L:PhD14}, or \cite[Lemma 5]{L:2013}, we have
\begin{displaymath}%\label{eq: P from my PhD}
P
=-\frac{7c}{960}[p']^2S(p)\1
+\frac{3}{20}p'\langle\vartheta'\rangle
+\frac{1}{5}p''\langle\vartheta\rangle
-(c+2)x\Gamma
\:,
\end{displaymath} 
where $S$ denotes the Schwarzian derivative.\footnote{The Schwarzian derivative of $f$ is defined by
$S(f)
=\frac{f'''}{f'}-\frac{3}{2}\left[\frac{f''}{f'}\right]^2$, for admissible $f$.}
%In the present case, $-\frac{7c}{48}p'''\1+3\langle\vartheta'\rangle=-\frac{13}{2}\1$, so
This yields the claimed value and implies $C=-\frac{c}{15}\:a\1$.
In order to determine the value of $D$, 
we use that \cite{LN:2017}
\begin{displaymath}%\label{eq: <vartheta Xs vartheta Xs>r}
\langle\vartheta(X_s)\vartheta(X_s)\rangle_r
=-\frac{c}{480}[p'(X_s)]^2S(p)(X_s)\1
-\frac{1}{10}\:p'(X_s)\:\langle\vartheta'(X_s)\rangle
+\frac{1}{5}\:p''(X_s)\langle\vartheta(X_s)\rangle\:.
\end{displaymath} 
The formula holds for arbitrary genus. 
In the present case, 
%$\frac{c}{48}p'''\1+\langle\vartheta'\rangle=(-c/2)\1$, so
\begin{displaymath}
D
%=\pi_{\C\1}\langle\vartheta(x)\vartheta(X_s)\rangle_r
=\pi_{\C\1}\langle\vartheta(X_s)\vartheta(X_s)\rangle_r
=\frac{c}{20}\pi_{\C\1}p'(X_s)
=\frac{ca}{5}\1
\end{displaymath}
where $\pi_{\C\1}$ denotes the projection onto $\C\1$.
% Since $\pi_{\C\1}R_*(x,X_s)=0$, we actually have
% \begin{displaymath}
% D
% =\pi_{\C\1}\langle\vartheta(x)\vartheta(X_s)\rangle
% \:.
% \end{displaymath}

\begin{proposition}\label{proposition: equivalence of two systems of ODE for g=1}
Let $\Sigma:y^2=p$ be the Riemann surface defined by eq.\ (\ref{def: polynomial p defining the torus as doubl cover of CP1}).
In the $(2,5)$ minimal model, the system (\ref{system: variation of zero-point function and A1 in terms of quotients of determinants})
of differential equations for $\1$ and $\A $ of Theorem \ref{Theorem: system of differential equations in algebraic formulation}
is equivalent to the system \cite{LN:2017}
\begin{equation}\label{system: variation of zero-point function and <vartheta>}
\begin{split}
\left(d-\frac{c}{8}\omega\right)\1
\=2\sum_{s=1}^3\frac{\xi_s}{p'(X_s)}\langle\vartheta(X_s)\rangle\:,\\
\left(d-\frac{c}{8}\:\omega\right)\langle\vartheta(x)\rangle
\=2\sum_{s=1}^3
\frac{\xi_s}{p'(X_s)}\langle\vartheta(X_s)\vartheta(x)\rangle
-\langle\vartheta(x)\rangle\frac{dp}{p}|_x
-\frac{c}{16}p'\:d\left(\frac{p'}{p}\right)|_x
\1
\:.
\end{split}
\end{equation}
\end{proposition}

\begin{proof}% of Proposition \ref{proposition: equivalence of two systems of ODE for g=1}
Cf.\ Appendix \ref{Appendix Section: Proof of Proposition {proposition: equivalence of two systems of ODE for g=1}}.
\end{proof}

% The equation for $\langle\vartheta(X_s)\rangle$ actually reads
% \begin{displaymath}
% \left(d_{X_s}-\frac{c}{8}\:\omega_s\right)\frac{\langle\vartheta(X_s)\rangle}{p'(X_s)}
% =\frac{\xi_s}{40}
% \left(
% \frac{7c}{8}\left[\frac{p''(X_s)}{p'(X_s)}\right]^2
% -\frac{13c}{2}\frac{a_0}{p'(X_s)}
% \right)\1
% +\frac{2}{5}\xi_s\frac{p''(X_s)}{p'(X_s)}\frac{\langle\vartheta(X_s)\rangle}{p'(X_s)}
% \:.
% \end{displaymath}

\subsection{The hypergeometric equation}

We define the torus as a double cover of $\CP^1$ defined by
\begin{displaymath}
y^2
=p(x)
=x(x-1)(x-\lambda)\:,
\end{displaymath}
% If $p$ is not monic the discussion remains the same.
where $\lambda\in\C$ is the squared Jacobi modulus.
The discriminant is
\begin{displaymath}
\Delta
=\left[\lambda(\lambda-1)\right]^2
\:, 
\end{displaymath}
and the connection 1-form equals
\begin{displaymath}
\omega
=\frac{2\lambda-1}{\lambda(\lambda-1)}d\lambda
=\pi\i E_2d\tau-6d(\log\ell)\:,
\end{displaymath}
where $\ell$ is the inverse length of the real period.
%\footnote{The equation is not modular under the full modular group but only under a subgroup of index $6$,
% which includes the transformation $\lambda\mapsto 1-\lambda$.}
The system of ODEs reads
\begin{equation}\label{system of eqs: ODEs for <1> and <vartheta>}
\begin{split}
\left(\frac{d}{d\lambda}-\frac{c}{8}\omega_{\lambda}\right)\1
\=2\frac{\langle\vartheta(\lambda)\rangle}{p'(\lambda)}\:,\\
%
% \left(\frac{d}{d\lambda}-\frac{c}{8}\frac{1}{\lambda(1-\lambda)}\right)\langle\vartheta(x)\rangle
% \=\frac{2}{p'(\lambda)}\langle\vartheta(\lambda)\vartheta(x)\rangle
% +\frac{\langle\vartheta(x)\rangle}{x-\lambda}
% -\frac{c}{16}\:\frac{p'}{(x-\lambda)^2}
% \1
% \:,\\
\left(\frac{d}{d\lambda}-\frac{c}{8}\omega_{\lambda}\right)\frac{\langle\vartheta(\lambda)\rangle}{p'(\lambda)}
\=
-\frac{7c}{480}\:
S(p)(\lambda)\1
+\frac{3}{10}\:\frac{\langle\vartheta'(\lambda)\rangle}{p'(\lambda)}
+\frac{2}{5}\:\frac{p''(\lambda)}{p'(\lambda)}\frac{\langle\vartheta(\lambda)\rangle}{p'(\lambda)}
\end{split}
\end{equation}
where $S$ is the Schwarzian derivative w.r.t.\ the coordinate $x$.

\begin{proposition}
Eq.\ (\ref{eq: 2nd order ODE for <1> in the (2,5) minimal model}) is equivalent to the hypergeometric ODE 
\begin{align}\label{eq: standard hypergeometric eq}
\frac{d^2}{d\lambda^2}f+p\frac{d}{d\lambda}f+qf
=0 
\end{align}
with rational coefficients
\begin{displaymath}
q=\frac{-\alpha\beta}{\lambda(1-\lambda)}\:,
\quad
p
%=\frac{\gamma-(\alpha+\beta+1)\lambda}{\lambda(1-\lambda)}
=\frac{\gamma}{\lambda}+\frac{\gamma-(\alpha+\beta+1)}{1-\lambda}\:,
\end{displaymath}
where
\begin{displaymath}
(\alpha,\beta;\gamma)
=\left(\frac{7}{10},\frac{11}{10};\frac{7}{5}\right)
\quad\text{or}\quad 
 \left(\frac{3}{10},-\frac{1}{10};\frac{2}{5}\right) 
\:.
\end{displaymath}
In particular, we have
\begin{align*}
\1_1
\=[\lambda(\lambda-1)]^{-1/30}\:_2F_1\left(\frac{7}{10},\frac{11}{10};\frac{7}{5};\lambda\right)\\
\1_2
\=[\lambda(\lambda-1)]^{11/30}\:_2F_1\left(\frac{3}{10},-\frac{1}{10};\frac{3}{5};\lambda\right)
\:.
\end{align*}
\end{proposition}

The latter relations seem to be new though they're closely related to Schwarz' work \cite{S:1873}
as will be indicated in Section \ref{section: The property of being algebraic}.

\begin{proof}
Let
\begin{displaymath}
g(\lambda)
:=\lambda(\lambda-1)
\:
\end{displaymath}
Let $\omega=\omega_{\lambda}d\lambda$ so
\begin{displaymath}
\omega_{\lambda}
=\frac{g'(\lambda)}{g(\lambda)} 
=\frac{1}{\lambda}-\frac{1}{1-\lambda}
\:.
\end{displaymath}
Let $\1$ be the $0$-point function w.r.t.\ the singular metric. 
We have 
\begin{align}\label{eqs: change from p to g}
g(\lambda)
=p'(\lambda)\:,
\quad
g'(\lambda)
=\frac{1}{2}p''(\lambda)\:,
\quad
g''(\lambda)
=2
\:.
\end{align}
For $\ell\in\R$, we have
\begin{displaymath}
\left(\frac{d}{d\lambda}
+\ell\:\omega_{\lambda}\right)
\left(\frac{d}{d\lambda}
+\ell\:\omega_{\lambda}\right)
=\frac{d^2}{d\lambda^2}
+2\ell\omega_{\lambda}\frac{d}{d\lambda}
+\frac{2\ell}{g}
+\ell(\ell-1)\omega_{\lambda}^2
\:.
\end{displaymath}
Moreover, for monic $p$ we have
\begin{align}\label{eq: <vartheta'> for g=1}
\langle\vartheta'(x)\rangle
=-(c/4)\1
\:, 
\end{align}
so from system (\ref{system of eqs: ODEs for <1> and <vartheta>}) follows that $\1$ lies in the kernel of the linear operator
\begin{displaymath}%\label{eq: 2nd order Diff op that kills <1>}
\frac{d^2}{d\lambda^2}
+\frac{3}{10}\omega_{\lambda}\frac{d}{d\lambda}
-\frac{3c}{160}\omega_{\lambda}^2
\:
\end{displaymath}
(the coefficient of the $1/g$ term equals $2\ell-\frac{7c}{240}p'''+\frac{3}{10}\langle\vartheta'\rangle/\1=0$).
For
\begin{displaymath}
F(\lambda)
=\ell\:\log g(\lambda)
%=\ell\:\int\frac{g'(\lambda)}{g(\lambda)}d\lambda
\:
\end{displaymath}
and for $m\in\N_0$ we have 
\begin{align}
e^{-F}\:\frac{d^m}{d\lambda^m}\:e^{F}
\=\left(\frac{d}{d\lambda}
+\ell\:\omega_{\lambda}\right)^m\:,
\label{eq: conjugation of m-th order derivative with exp F}
\end{align}
where the $m$-th power refers to the composition of differential operators. 
Now let $\ell=k-\frac{c}{8}$.
By eqs\ (\ref{eqs: change from p to g}), (\ref{eq: <vartheta'> for g=1}) 
and (\ref{eq: conjugation of m-th order derivative with exp F}),
the system of ODEs (\ref{system of eqs: ODEs for <1> and <vartheta>})
is equivalent to the system 
\begin{equation}\label{system of eqs: ODEs for <1>*k and <vartheta>*k}
\begin{split}
\left(\frac{d}{d\lambda}-k\omega_{\lambda}\right)\1^*_k
\=2\frac{\langle\vartheta(\lambda)\rangle^*_k}{g(\lambda)}\:,\\
\left(\frac{d}{d\lambda}-k\omega_{\lambda}\right)\frac{\langle\vartheta(\lambda)\rangle^*_k}{g(\lambda)}
% \=
% -\frac{7c}{80}\:
% \left(\frac{1}{g(\lambda)}-\left[\omega_{\lambda}\right]^2\right)\1^*_k
% +\frac{4}{5}\:\omega_{\lambda}\frac{\langle\vartheta(\lambda)\rangle^*_k}{g(\lambda)}
% -\frac{3c}{40}\:\frac{\1^*_k}{g(\lambda)}\\
\=
\left(\frac{7c}{80}\omega_{\lambda}^2
-\frac{13c}{80}\frac{1}{g(\lambda)}\right)\1^*_k
+\frac{4}{5}\:\omega_{\lambda}\frac{\langle\vartheta(\lambda)\rangle^*_k}{g(\lambda)}
% \\
% \=
% \:
% \left(\frac{7c}{80}\left[\omega_{\lambda}\right]^2
% -\frac{13c}{80}\frac{1}{g(\lambda)}\right)\1^*_k
% +\frac{2}{5}\:\omega_{\lambda}\frac{d}{d\lambda}\1^*_k
\end{split}
\end{equation}
for the amended functions 
\begin{displaymath}
\1^*_k
:=e^{F(\lambda)}\1\:,\quad
\langle\vartheta(\lambda)\rangle^*_k
:=e^{F(\lambda)}\langle\vartheta(\lambda)\rangle
\:.
\end{displaymath}
From the system (\ref{system of eqs: ODEs for <1>*k and <vartheta>*k}) follows that $\1^*_k$ lies in the kernel of the linear operator
\begin{align*}
\frac{d^2}{d\lambda^2}
-\left(\frac{4}{5}+2k\right)\omega_{\lambda}\frac{d}{d\lambda}
+\left(\frac{4k}{5}-\frac{7c}{40}+k(k+1)\right)\omega_{\lambda}^2
+\left(\frac{13c}{40}-2k\right)\frac{1}{g}
\:.
\end{align*} 
Only those values of $k$ are allowed for which the second order poles drop out,
% \begin{displaymath}
% k(k+1)=\frac{7c}{40}-\frac{4k}{5}
% \quad\Leftrightarrow\quad
% k^2+\frac{9}{5}k-\frac{7c}{40}=0\:.
% \end{displaymath}
these are
\begin{displaymath}
k_{1/2}
% =-\frac{9}{10}\pm\sqrt{\frac{81}{100}-\frac{77}{100}}
% =-\frac{9}{10}\pm\frac{2}{10}
=-\frac{7}{10},-\frac{11}{10}
\:. 
\end{displaymath}
The corresponding function $\1^*_k$ solves 
\begin{displaymath}
\left[
\frac{d^2}{d\lambda^2}
-\left(\frac{4}{5}+2k\right)\omega_{\lambda}\frac{d}{d\lambda}
+\left(\frac{13c}{40}-2k\right)\frac{1}{g(\lambda)}\right]\1^*_k
=0
\:.
\end{displaymath} 
Comparison with hypergeometric differential equation (\ref{eq: standard hypergeometric eq}) yields
\begin{displaymath}
(\alpha\beta,\:\gamma)
=\left(\frac{13c}{40}-2k,\:-2k-\frac{4}{5}\right)
%=-\frac{143}{100}-2k
=
\begin{cases}
\left(-\frac{3}{100},\:\frac{3}{5}\right)&\quad\text{for $k=-\frac{7}{10}$}\\
\:\left(\frac{77}{100},\:\frac{7}{5}\right)&\quad\text{for $k=-\frac{11}{10}$}
\end{cases}
\end{displaymath}
Moreover,
\begin{displaymath}
\alpha+\beta
=2\gamma-1
=
\begin{cases}
\frac{1}{5}&\quad\text{for $k=-\frac{7}{10}$}\\
\frac{9}{5}&\quad\text{for $k=-\frac{11}{10}$}
\end{cases}
\end{displaymath}
This yields the propositioned values for $\alpha,\beta$ and gives
\begin{align*}
\1^*_{-7/10}
\=\:_2F_1\left(\frac{7}{10},\frac{11}{10};\frac{7}{5};\lambda\right)\\
\1^*_{-11/10}
\=\:_2F_1\left(\frac{3}{10},-\frac{1}{10};\frac{3}{5};\lambda\right)
\end{align*}
To make the identification with the Rogers-Ramanujan functions $H$ and $G$ defined by eqs (\ref{defs: Rogers-Ramanujan functions}),
we recall eq.\ (\ref{eq: identity relating 0-point function in flat metric to that in singular metric}).
Thus
\begin{itemize}
\item for $k=-\frac{7}{10}$, we have 
$\ell
%=\frac{11}{20}-\frac{7}{10}
=-\frac{3}{20}$ and
\begin{displaymath}
% \1^*_{-7/10}
% =
\:_2F_1\left(\frac{7}{10},\frac{11}{10};\frac{7}{5};\lambda\right)
=g(\lambda)^{-\frac{3}{20}}\onesing
=\Delta^{-\frac{3}{40}}\onesing
=\Delta^{-\frac{3}{40}-\frac{c}{48}}H
=\Delta^{\frac{1}{60}}H
\end{displaymath}
lies in the kernel of
\begin{displaymath}
\frac{d^2}{d\lambda^2}
+\frac{3}{5}\frac{1}{\lambda\:(\lambda-1)}\frac{d}{d\lambda}+\frac{3}{100}\frac{1}{\lambda\:(\lambda-1)}
\:.
\end{displaymath}
\item
For $k=-\frac{11}{10}$, we have 
$\ell
%=\frac{11}{20}-\frac{11}{10}
=-\frac{11}{20}$ and
\begin{displaymath}
% \1^*_{-11/10}
% =
\:_2F_1\left(\frac{3}{10},-\frac{1}{10};\frac{3}{5};\lambda\right)
=g(\lambda)^{-\frac{11}{20}}\onesing
=\Delta^{-\frac{11}{40}}\onesing
=\Delta^{\frac{c}{16}-\frac{c}{48}}G
%=\Delta^{\frac{c}{24}}G
=\Delta^{-\frac{11}{60}}G
\end{displaymath}
lies in the kernel of
\begin{displaymath}
\frac{d^2}{d\lambda^2}
+\frac{7}{5}\frac{1}{\lambda\:(\lambda-1)}\frac{d}{d\lambda}-\frac{77}{100}\frac{1}{\lambda\:(\lambda-1)}
\:.
\end{displaymath}
\end{itemize}
This completes the proof.
\end{proof}

\section{Algebraicity of the Rogers-Ramanujan characters}\label{section: The property of being algebraic}

Besides the analytic approach, there is an algebraic approach to the characters.
% This is due to the fact that $\1_1,\1_2$,
% rather than being modular on the full modular group, are modular on a subgroup of $\Gamma_1$:
% For the generators $S,T$ of $\Gamma_1$ we have \cite{CJMV-Z:2007}
% \begin{displaymath}
% T\1_1
% ={\zeta_{60}}^{11}\1_1\:,
% \quad
% T\1_2
% ={\zeta_{60}}^{-1}\1_2\:,
% \end{displaymath}
% while under the operation of $S$, $\1_1,\1_2$ transform into linear combinations of one another \cite{CJMV-Z:2007},
% \begin{displaymath}
% S\binom{\1_1}{\1_2}
% =\frac{2}{\sqrt{5}}
% \begin{pmatrix}
% \sin\frac{\pi}{5}&-\sin\frac{2\pi}{5}\\  
% \sin\frac{2\pi}{5}& \sin\frac{\pi}{5}     
% \end{pmatrix}
% \binom{\1_1}{\1_2}\:.
% \end{displaymath}
% However, $\1_1,\1_2$ are modular under a subgroup of $\Gamma_1$ of finite index. 

%\subsection{Motivation}

From general theory, we know that any two modular functions are algebraically dependent \cite[Propos 3, p. 12]{Z:1-2-3},
so 

\begin{proposition}
The Rogers Ramanujan functions are algebraic in the $j$-invariant.
\end{proposition}

% Let $\Gamma\subset SL(2,\R)$ be a discrete subgroup for which $\Gamma\setminus\mathfrak{h}$ has finite volume $V$.\\
% Any \textbf{three} modular forms resp.\ any \textbf{two} modular functions on $\Gamma$ are algebraically dependent.
% 
% \begin{proof}
% Suppose $f,g,h$ are algebraically independent modular forms of weight $>0$.
% There exists $a>0$ such that
% \begin{displaymath}
% ak^2
% <\sharp\{\text{monomials in $f,g,h$ of total weight $k$}\}
% \leq\dim M_k(\Gamma)
% \leq k\frac{V}{4\pi}+1\:.
% \end{displaymath}
% contradiction for $k\rechts\infty$.
% 
% It follows that any two modular functions on $\Gamma$ are algebraically dependent
% (since every modular function is a quotient of two modular forms).
% \end{proof}

We are interested in generalising this result to higher genus.
As a preparation, the specific algebraic equations for the Rogers-Ramanujan functions will be discussed. 
% On the other hand, there are the classical arguments due to Schwarz \cite{S:1873} and Klein \cite{K:1884}.
% Using differential equations, some of these arguments may generalise to higher genus:
% We expect that the $0$-point functions of the $(2,5)$ minimal model continue to be algebraic for $g\geq 2$.
% 
% 
% The Rogers-Ramanujan functions 
% satisfy the hypergeometric differential equation for a specific triple of parameters $\alpha,\beta,\gamma$,
% and are algebraic by the Schwarz' classification. 
% Thus by Riemann, they define a finite covering of $\CP^1$.
% This leads to the consideration of platonic solids.

\subsection{Schwarz' list}

A necessary condition for the general solution of the hypergeometric differential equation (\ref{eq: standard hypergeometric eq}) to be algebraic in $\lambda$ is that $\alpha,\beta,\gamma\in\Q$
(Kummer), which we will assume in the following. %For the rest of the talk, we will make this assumption on the parameters.

\begin{proposition}
Let $f_1,f_2$ be solutions of (\ref{eq: standard hypergeometric eq}),
for some choice of $\alpha,\beta,\gamma\in\Q$,
such that 
\begin{displaymath}
s=f_1/f_2 
\end{displaymath} 
is algebraic.
Then $f_1,f_2$ are themselves algebraic.
\end{proposition}

A particularly neat argument is due to Heine \cite[and reference therein]{S:1873}.

\begin{proof}
Let $W =f_1'f_2-f_2'f_1$ be the Wronskian.
Since $s$ is algebraic and $s'=W/f_2^2$, it suffices to show that $W$ is algebraic:
We have
\begin{displaymath}
W'
=f_1''f_2-f_2''f_1
=-pW
\end{displaymath}
by eq.\ (\ref{eq: standard hypergeometric eq}),
so for $A,B$ such that
\begin{displaymath}
p(\lambda)
%=\frac{C}{z}+\frac{A+B+1-C}{z-1}
=-\frac{A}{\lambda}-\frac{B}{\lambda-1}\:, 
\end{displaymath}
we have
\begin{displaymath}
W
\sim\exp\left(-\int p\:d\lambda\right)
=\lambda^{A}(\lambda-1)^{B}\:.
\end{displaymath}
By assumption $A,B\in\Q$.
\end{proof}

Given two independent algebraic solutions $f_1,f_2$ to (\ref{eq: standard hypergeometric eq}),
their quotient 
\begin{displaymath}
s=f_1/f_2 
\end{displaymath}
solves a third order differential equation in $\lambda$ \cite[p.\ 299]{S:1873},
which involves the Schwarzian derivative.
% \footnote{
% $s=f_1/f_2$ satisfies $S(s,\lambda)=\frac{s'''}{s'}-\frac{3}{2}\left[\frac{s''}{s'}\right]^2=2q-\frac{1}{2}p^2-p'=F(\lambda)$, for $p,q$ as in Sect.\ \ref{section: some ODEs}}
By linearity of (\ref{eq: standard hypergeometric eq}), the space of solutions is invariant under M\"obius transformations.
$s$ defines a map
\begin{equation*}%\label{multi-valued map from P1 to P1}
\begin{split}
s:\quad\CP^1\setminus\{0,1,\infty\}\rechts&\CP^1\left(\cong S^2\right)\\
\lambda\:\mapsto&\:(f_1:f_2) \:. 
\end{split}
\end{equation*}
Suppose $f_1^{\R},f_2^{\R}$ are real on $(0,1)$ 
% For real parameters $\alpha,\beta,\gamma$, and for real $\lambda$ (outside of the singularities of $p,q$),
% and more generally for a real linear differential equation in $\lambda$, 
% there is always, locally in $\lambda$, a basis of real solutions. 
% Clearly if $f$ is a complex solution then so is its complex conjugate, 
% and by linearity also its real part solves the equation etc.
% Alternatively, given a real solution, use the real differential equation 
% to compute its Taylor expansion.  
% The hypergeometric differential equation (\ref{eq: standard hypergeometric eq}) is real 
% and so are its solutions (outside of the singularities of $p,q$).
% Indeed, if $f$ is a solution to (\ref{eq: standard hypergeometric eq})
% and $f,f'$ are real on $(0,1)$ then 
% $f''=-pf'-qf$ is real on that interval since the r.h.s.\ is by assumption.
and their quotient $s^{\R}=f_1^{\R}/f_2^{\R}$ maps the interval $(0,1)$
onto a segment $I_{(0,1)}$ of $\RP^1\cong S^1$.
Via an analytic extension to $\mathfrak{h}$, $s^{\R}$ can be extended to the intervals $(-\infty,0)$ and $(1,\infty)$.
For $\eps>0$, the interval $(-\eps,\eps)$ is mapped to two arcs forming some angle.
Together, the images of $(0,1)$, $(1,\infty)$ and $(-\infty,0)$ form a triangle in $\CP^1$. % \footnote{
% That the arcs do fit is a local statement for every vertex,
% which is verified by considering the respective Frobenius expansion of $f_1/f_2$.
% So $f_1/f_2$ has a well-defined limit in $\CP^1$ when $X\rechts 0^+$.} For a suitable metric, the arcs bounding the triangles are geodesic.
In the elliptic case (angular sum $>180^{\circ}$), the triangle is conformally equivalent to a spherical triangle on $S^2$ whose edges are formed by arcs of great circles. 

By crossing any of the intervals $(1,\infty)$ $(0,1)$, or $(-\infty,0)$, 
$s^{\R}$ can be further continued to $\H^-$. 
We have a correspondence between reflection symmetry w.r.t.\ the real line in the $\lambda$-plane
and circle inversion w.r.t.\ the respective triangle edge in $\CP^1$.

Analytic continuating along paths circling the singularities in any order
may in general produce an infinite number of triangles in $\CP^1$. 
The number is finite iff the quotient of solutions is finite \cite[Sect.\ 20]{R:1851}.
For angle sums $\leq 180^{\circ}$ finite coverings are impossible.

Thus the problem is transformed into sorting out all spherical triangles whose symmetric and congruent repetitions lead to a finite number only of triangles of different shape and position. 

A necessary condition for a spherical shape and its symmetric and congruent repetions to form a closed Riemann surface
is that the edges lie in planes which are symmetry planes of a regular polytope. 

For the spherical triangles, this leads to a finite list of triples of angles that correspond to platonic solids. 

The Rogers-Ramanujan functions feature as the most symmetric case (no.\ XI, i.e.\ all three angles equal $2\pi/5$) 
in the list of Schwarz \cite{S:1873}.

\hspace{1cm}

% \begin{align*}
% &\CP^1\longrightarrow\Gamma\setminus\CP^1\\
% &\:\downarrow x\quad\nearrow s\\
% \H^+\cup\H^-&\cup\{\infty\}
% \end{align*}
% We have $\H^+\cup\H^-\cup\{\infty\}$ with variable $\lambda$ on which $s$ defines a multi-valued, locally one-valued function (like $\sqrt{z}$)
% to $\Gamma\setminus\CP^1$, where $\Gamma$ is the group of reflections (circle inversions) with composition.
% Under the operation of $\Gamma$, all spherical triangles on $\CP^1$ are identified with the one formed from $(-\infty,0)$, $(0,1)$ and $(1,\infty)$.
% This quotient is covered by a $\CP^1$ on which a point is a path modulo deformations which do not pass through any of the singularities $0,1$,
% thus an element in $\pi^1(\C\setminus\{0,1\})$, where we require the covering map to have at most finiteky many poles.
% Morweover, we have a map $x$ from this $\CP^1$ to $\H^+\cup\H^-\cup\{\infty\}$ which forgets the specific path and returns only the endpoint.
% In this situation, $s$ is algebraic in $\lambda$.

The compactified fundamental domain $\overline{\Gamma_1\setminus\mathfrak{h}}=\Gamma_1\setminus(\mathfrak{h}\cup\Q\cup\{\infty\})$
of $\Gamma_1$ \cite{Z:1-2-3} is conformally equivalent to $\CP^1$. 
The $j$-invariant defines a Hauptmodul for $\Gamma_1$.
%($j$ defines an isomorphism $\overline{\Gamma_1\setminus\mathfrak{h}}\rechts\CP^1$ and its simple pole lies at $i\infty$).

On the other hand, the modular curve of the principal congruence subgroup $\Gamma(N)$ has $g=0$ iff $1\leq N\leq 5$.
% Indeed, the fundamental domain of $\Gamma(N)$ is required to define a finite covering of the fundamental domain of $\Gamma_1$
% for which the orbifold points $i$ and $\rho$ become interior points (and thus no orbifold points?). 
% The map $\Gamma(N)\setminus\mathfrak{h}\rechts\CP^1$ is conformal outside of the cusps. $i,\rho$ are not cusps, but $i\infty$ is a cusp.
%Thus the angles $\frac{\pi}{2}$ and $\frac{\pi}{3}$ at $i$ and $\rho$, respectively, are preserved under this identification,
%but the angle at $i\infty$ is not.
For $N\geq 2$, the map $\Gamma(N)\setminus\mathfrak{h}\rechts\CP^1$ is conformal outside of the cusps.
Thus the angle $\frac{\pi}{3}$ at $\rho$ is preserved under this map.
Since $N$ copies of the fundamental domain of $\Gamma_1$ meet in the cusp at $i\infty$,
the compatified fundamental domain of $\Gamma(N)$ defines a finite covering iff 
$\frac{2\pi}{N}+\frac{2\pi}{3}>\pi$, or equivalently $N<6$.

For $N=5$, the angle at the image of $i\infty$ equals $72^{\circ}$.
The modular curve $\Gamma(5)\setminus(\mathfrak{h}\cup\Q\cup\{\infty\})$ has the symmetry of an icosahedron.
% $i$ and $\rho$ become edge- and face centers, respectively.
% Every face of the icosahedron is subdivided into three triangles which have a common vertex at the centroid.
By modularity on $\Gamma(5)$, 
$r(\tau)=\1_1/\1_2$ defines a map
\begin{displaymath}
\Gamma(5)\setminus\mathfrak{h}\rechts\CP^1
\end{displaymath}
$r(\tau)$ is a Hauptmodul for $\Gamma(5)$.
We have $[\Gamma_1:\Gamma(5)]=120$ \cite{Gu:1962},
so by eq.\ (\ref{eq: fundamental domain of finite index subgroup of Gamma1}), 
the fundamental domain of $\Gamma(5)$ defines an $120$-fold covering of $\CP^1$,
and $r$ and $j$ are rational functions of one another.

% While $a,b,c,d$ run over all integers with $ad-bc=1$, $r\left(\frac{ai+b}{ci+d}\right)$
% runs over the edge points and $r\left(\frac{a\rho+b}{c\rho+d}\right)$ runs over the face points.

\subsection{Klein's invariants}

Felix Klein reverses the order of arguments used by Schwarz. 

\begin{theorem}\label{theorem: Klein}(Felix Klein)
The icosahedral irrationality 
\begin{equation}\label{eq: Klein's algebraic fct related to the icosahedron}
% \frac{\theta_{5,2}(\tau)}{\theta_{5,1}(\tau)}
% =
q^{1/5}
\frac{\sum_{n=-\infty}^{\infty}(-1)^nq^{\frac{5n^2+3n}{2}}}{\sum_{n=-\infty}^{\infty}(-1)^nq^{\frac{5n^2+n}{2}}}
\:,
\end{equation}
% $e^{2\pi i\tau/5}\frac{\sum_{n=-\infty}^{\infty}(-1)^n(e^{\pi i\tau})^{5n^2+3n}}
% {\sum_{n=-\infty}^{\infty}(-1)^n(e^{\pi i\tau})^{5n^2+n}}$ 
\cite[Part I, eq.\ (20), p.\ 146]{K:1884} is algebraic.  
\end{theorem}

By the Jacobi triple product identity, we have \cite{Z:1987}
\begin{align*}
\prod_{n=0,\pm 2\:\text{mod}\:5}(1-q^n)
\=\sum_{n=-\infty}^{\infty}(-1)^nq^{\frac{5n^2+n}{2}}
\\
\prod_{n=0,\pm 1\:\text{mod}\:5}(1-q^n)
\=\sum_{n=-\infty}^{\infty}(-1)^nq^{\frac{5n^2+3n}{2}}
\end{align*}
so the function (\ref{eq: Klein's algebraic fct related to the icosahedron}) is $z=r(\tau)=\1_1/\1_2$.
Klein shows that $\tilde{z}=z^5$ satisfies the icosahedral equation
\begin{displaymath}
%  F(z)^3+V(z)^5J(z)
%  =
(\tilde{z}^4-228\tilde{z}^3+494\tilde{z}^2+228\tilde{z}+1)^3+\tilde{z}(\tilde{z}^2+11\tilde{z}-1)^5j(\tau)
=0\:,
\end{displaymath}
where $j(\tau)$ is the modular $j$-invariant. The icosahedral equation is the minimal polynomial of $\tilde{z}$ over $\Q(j(\tau))$.
It yields an expression of $j(\tau)$ as a rational function of $r(\tau)$,
and $r(\tau)^5$ defines a $12$-fold covering of $\CP^1$.

A modern treatment of Klein's invariants can be found in \cite{D:2005}.

\hspace{0.5cm}

$r$ is determined up to linear fractional transformation,
so its Schwarzian derivative is unique, and we obtain a third order ODE for $r$ as a function of $j$.
$\1_1$ and $\1_2$ define projective coordinates 
or elements in the two-dimensional space of global rational sections in the sheaf $\O(1)$.
Thus they define two solutions of a linear $2$nd order ODE in $j$ with rational coefficients.

As an algebraic function of 
\begin{displaymath}
j
=2^8\frac{(1-\lambda(1-\lambda))^3}{\lambda^2(1-\lambda)^2} 
\:,
\end{displaymath}
$r$ is also an algebraic function of $\lambda$. This leads to our corresponding ODEs w.r.t\ $\lambda$.

\pagebreak

\appendix

\section{Appendix}

\subsection{Completion of the proof of Proposition \ref{propos: determinants in terms of a,b and da, db}}

\subsubsection{Some useful formulae}

We have
\begin{align*}
a=\overline{X_1X_2}\:,
\hspace{1cm}
d a
\=d(\overline{X_1X_2})\\
\=\xi_1X_2+\xi_1X_3+\xi_2X_1+\xi_2X_3+\xi_3X_1+\xi_3X_2
=\overline{\xi_1X_2}\\
% \=-\xi_1X_1-\xi_2X_2-\xi_3X_3
% =-\overline{\xi_1X_1}\\
b=-X_1X_2X_3\:,
\hspace{1cm}
d b
\=-d(X_1X_2X_3)\\
\=-\xi_1X_2X_3-\xi_2X_1X_3-\xi_3X_1X_2
=-\overline{\xi_1X_2X_3}\:. 
\end{align*}
Since $\overline{X_1}=0$, we have 
\begin{align}
(\overline{X_1X_2})^2
%\=\overline{X_1^2X_2^2}+2\overline{X_1^2X_2X_3}
=\overline{X_1^2X_2^2}+2X_1X_2X_3\cdot\overline{X_1}
=\overline{X_1^2X_2^2}\:,\label{(X1 X2) squared} 
\end{align}
and
\begin{align}
\overline{\xi_1X_1}
\=-\overline{\xi_1X_2}\:,\label{(xi1 X1)}\\
\overline{\xi_1X_1^2}
\=-\overline{\xi_1X_1X_2}\nn\\
\=-\xi_1X_1(X_2+X_3)+\cyc
=-\overline{X_1X_2}\cdot\overline{\xi_1}+\overline{\xi_1X_2X_3}
=\overline{\xi_1X_2X_3}\:.\label{(xi1 X1 X2) and (xi1 X2 X3)}
\end{align}
Moreover,
\begin{align}
\overline{X_1^3}
\=X_1(X_2+X_3)^2+\cyc
=\overline{X_1X_2^2}+6X_1X_2X_3
=3X_1X_2X_3\:,\label{(X1 to the cube)}
\end{align}
since
\begin{align*}
\overline{X_1X_2^2}
\=-X_1X_2(X_1+X_3)-X_1X_2(X_2+X_3)+\cyc
=-6X_1X_2X_3-\overline{X_1^2X_2}
=-3X_1X_2X_3\:,
\end{align*}
and we have
\begin{align}
\overline{\xi_1X_1^3}
\=\xi_1X_1(X_2+X_3)^2+\cyc
=\overline{\xi_1X_1X_2^2}+2X_1X_2X_3\cdot\overline{\xi_1}
=\overline{\xi_1X_1X_2^2}\nn\\
\overline{X_1^2}
\=\overline{\xi_1X_1^2X_2^3}+\overline{\xi_1X_1^2X_2X_3^2}+\overline{\xi_1X_2^3}-X_1(X_2+X_3)+\cyc
=-2\:\overline{X_1X_2}
\:.\label{(X1 squared)}
\end{align}

\subsubsection{Proof of eq.\ (\ref{det Xi(3,0)V in terms of da and db})}
\label{Appendix: Proof of eq. det Xi(3,0)V in terms of da and db}

It remains to show eq.\ (\ref{det Xi(3,0)V in terms of da and db})
for general deformations $\xi_i=dX_i$, assuming that $\overline{X_1}=0$, eq.\ (\ref{eq: vanishing condition on X}).

Let $\alpha,\beta\in\Q$.
On the one hand,
\begin{align*}
\alpha\:a\:d b+\beta\:b\:d a
\=-\alpha\overline{X_1X_2}\cdot\overline{\xi_1X_2X_3}-\beta\:X_1X_2X_3\cdot\overline{\xi_1X_2}\\
\=-(\alpha+\beta)\:X_1X_2X_3\cdot\overline{\xi_1X_2}
-\overline{\xi_1X_1^2X_2^3}+\overline{\xi_1X_1^2X_2X_3^2}+\overline{\xi_1X_2^3}\alpha\:\overline{\xi_1X_2^2X_3^2}\:.
\end{align*}
On the other hand, 
\begin{align*}
\det\Xi_{3,0}\det V_3
\=\det
\begin{pmatrix}
\xi_1&\xi_2&\xi_3\\
X_1&X_2&X_3\\
1&1&1
\end{pmatrix} 
\begin{pmatrix}
1&X_1&X_1^2\\
1&X_2&X_2^2\\
1&X_3&X_3^2
\end{pmatrix}\nn\\
\=\det
\begin{pmatrix}
0&\overline{\xi_1X_1}&\overline{\xi_1X_1^2}\\
0&\overline{\xi_1X_1^2X_2^3}+\overline{\xi_1X_1^2X_2X_3^2}+\overline{\xi_1X_2^3}\overline{X_1^2}&\overline{X_1^3}\\
3&0&\overline{X_1^2}\\
\end{pmatrix}
=3
\left(\overline{X_1^3}\cdot\overline{\xi_1X_1}-\overline{X_1^2}\cdot\overline{\xi_1X_1^2}\right)\:.
\end{align*} 
Eqs (\ref{(xi1 X1)}), (\ref{(xi1 X1 X2) and (xi1 X2 X3)}), (\ref{(X1 to the cube)}), and (\ref{(X1 squared)}) yield
\begin{align*}
\det\Xi_{3,0}\det V_3
\=3
\left(-3X_1X_2X_3\cdot\overline{\xi_1X_2}+2\:\overline{X_1X_2}\cdot\overline{\xi_1X_2X_3}\right)\nn\\
\=3
\left(-3X_1X_2X_3\cdot\overline{\xi_1X_2}+2\:\overline{\xi_1X_2^2X_3^2}+2X_1X_2X_3\cdot\overline{\xi_1X_2}\right)\nn\\
\=-3X_1X_2X_3\cdot\overline{\xi_1X_2}+6\:\overline{\xi_1X_2^2X_3^2}\:.
\end{align*}
We conclude that $\alpha=-6$, $\alpha+\beta=3$, so $\beta=9$. 

\subsubsection{Proof of Eq.\ (\ref{det Xi(3,1)V in terms of da and db})}
\label{Appendix: Proof of eq. det Xi(3,1)V in terms of da and db}

It remains to show eq.\ (\ref{det Xi(3,1)V in terms of da and db}) for general deformations $\xi_i=dX_i$, 
assuming that $\overline{X_1}=0$, eq.\ (\ref{eq: vanishing condition on X}).

Let $\alpha,\beta\in\Q$.
On the one hand, we have by eq.\ (\ref{(X1 X2) squared})
\begin{displaymath}
\alpha a^2da+\beta b\:db
%\=\alpha (\overline{X_1X_2})^2\cdot\overline{\xi_1X_2}+\beta X_1X_2X_3\cdot\overline{\xi_1X_2X_3}\\
=\alpha\:\overline{X_1^2X_2^2}\cdot\overline{\xi_1X_2}+\beta X_1X_2X_3\cdot\overline{\xi_1X_2X_3}\:.
\end{displaymath}
On the other hand,
\begin{align*}
\det\Xi_{3,1}\det V_3
\=\det
\begin{pmatrix}
\xi_1X_1&\xi_2X_2&\xi_3X_3\\
X_1&X_2&X_3\\
1&1&1
\end{pmatrix} 
\begin{pmatrix}
1&X_1&X_1^2\\
1&X_2&X_2^2\\
1&X_3&X_3^2
\end{pmatrix}\nn\\
\=\det
\begin{pmatrix}
\overline{\xi_1X_1}&\overline{\xi_1X_1^2}&\overline{\xi_1X_1^3}\\
0&\overline{X_1^2}&\overline{X_1^3}\\
3&0&\overline{X_1^2}\\
\end{pmatrix}\\
\=3
\left(\overline{X_1^3}\cdot\overline{\xi_1X_1^2}-\overline{X_1^2}\cdot\overline{\xi_1X_1^3}\right)
+\left(\overline{X_1^2}\right)^2\cdot\overline{\xi_1X_1}\nn\:,
\end{align*} 
where by eq.\ (\ref{(X1 X2) squared}),
\begin{displaymath}
\left(\overline{X_1^2}\right)^2
=4\left(\overline{X_1X_2}\right)^2
=4\:\overline{X_1^2X_2^2}\:,
\end{displaymath}
and eqs (\ref{(xi1 X1)}), (\ref{(xi1 X1 X2) and (xi1 X2 X3)}), (\ref{(X1 to the cube)}), and (\ref{(X1 squared)}) apply.
Moreover,
\begin{align*}
\overline{X_1X_2}\cdot\overline{\xi_1X_1X_2^2} 
\=(X_1X_2+X_1X_3+X_2X_3)(\xi_1X_1X_2^2+\xi_1X_1X_3^2+\cyc)\\
\=X_1^2X_2^2\cdot\xi_1X_2+X_1^2X_3^2\cdot\xi_1X_3+\cyc\\
\+X_1X_2\cdot\xi_1X_1X_3^2+X_1X_3\cdot\xi_1X_1X_2^2+\cyc\\
\+X_2X_3\cdot(\xi_1X_1X_2^2+\xi_1X_1X_3^2)+\cyc\\
\=\overline{X_1^2X_2^2}\cdot\overline{\xi_1X_2}
+X_1X_2X_3\cdot\overline{\xi_1X_1X_2} 
+X_1X_2X_3\cdot\overline{\xi_1X_2^2}\\
\=\overline{X_1^2X_2^2}\cdot\overline{\xi_1X_2}\:,
\end{align*}
by eq.\ (\ref{(xi1 X1 X2) and (xi1 X2 X3)}) and
\begin{align*}
\overline{\xi_1X_2^2}
\=-\xi_1X_2(X_1+X_3)-\xi_1(X_1+X_2)X_3+\cyc\\
\=-\overline{\xi_1X_1X_2}-2\:\overline{\xi_1X_2X_3}
=\overline{\xi_1X_2X_3}\:. 
\end{align*}
We conclude that
\begin{align*}
\det\Xi_{3,1}\det V_3
\=9X_1X_2X_3\cdot\overline{\xi_1X_2X_3}
+6\overline{X_1X_2}\cdot\overline{\xi_1X_1X_2^2}
-4\overline{X_1^2X_2^2}\cdot\overline{\xi_1X_2}\nn\\
\=9X_1X_2X_3\cdot\overline{\xi_1X_2X_3}
+2\overline{X_1^2X_2^2}\cdot\overline{\xi_1X_2}\:,
\end{align*}
and so $\alpha=2$, $\beta=9$, as required.
This completes the proof.

\subsection{Completion of the proof of Theorem \ref{Theorem: system of differential equations in algebraic formulation}
(Section \ref{Subsection: variation of the conformal structure for g=1})}\label{Appendix Section: Completion of the proof of Theorem on DE in algebraic formulation}

It remains to show that
\begin{align*}
-\frac{\Theta(X_1)(\xi_2X_3+\xi_3X_2)}{(X_1-X_2)(X_3-X_1)}+\cyc
\=-\frac{2}{3}ca_2\1\frac{\det\Xi_{3,0}}{\det V_3}
-2\A \frac{\det\Xi_{3,1}}{\det V_3}\:.
\end{align*}
We have
\begin{align*}
\xi_2X_3+\xi_3X_2
\=(\xi_2+\xi_3)(X_2+X_3)-(\xi_2X_2+\xi_3X_3)\\
\=\xi_1X_1-(\xi_2X_2+\xi_3X_3)\\
\=2\xi_1X_1-\overline{\xi_1X_1}\:.
\end{align*}
It follows that
\begin{displaymath}
-\frac{\Theta(X_1)(\xi_2X_3+\xi_3X_2)}{(X_1-X_2)(X_3-X_1)}+\cyc
=\frac{8c\1\:\xi_1X_1^2-2\A \xi_1X_1}{(X_1-X_2)(X_3-X_1)}
+\cyc\:,
\end{displaymath}
since $\overline{\xi_1X_1}$ is symmetric and we have eqs (\ref{cyclic sum of 1 over Ni is zero}) and (\ref{cyclic sum of Xi over Ni is zero}). 
Moreover, we have eq.\ (\ref{eq for xi1X2X3 over N1}).
From eq.\ (\ref{X1 squared})  follows eq.\ (\ref{eq for xi1X2X3 over N1}),
and the proof of Theorem \ref{Theorem: system of differential equations in algebraic formulation} is complete.

\subsection{Proof of Proposition \ref{proposition: equivalence of two systems of ODE for g=1}}\label{Appendix Section: Proof of Proposition {proposition: equivalence of two systems of ODE for g=1}}

By eqs (\ref{def: vartheta}) and (\ref{eq: omega as difference of xi's over difference of Xi's}), we have
\begin{displaymath}
2\sum_{s=1}^3
\frac{\xi_s}{p'(X_s)}\langle\vartheta(X_s)\rangle
=2\sum_{s=1}^3
\frac{\xi_s}{p'(X_s)}\left(-cX_s\1+\frac{\A }{4}\right)
=-\frac{c}{6}\omega\1-\frac{\A }{8}\frac{\det\Xi_{3,0}}{\det V_3}\:,
\end{displaymath}
so the ODE for $\1$ in system (\ref{system: variation of zero-point function and <vartheta>})
is equivalent to that in system (\ref{system: variation of zero-point function and A1 in terms of quotients of determinants}).
We address the ODE for $\langle\vartheta(x)\rangle$.
For $x\rechts\infty$, we have eq.\ (\ref{incomplete eq. for <vartheta(x)vartheta(Xs)> for n=3})
which relates the ODE for $\langle\vartheta(x)\rangle$
with that for $\1$.
The resulting equation is necessarily compatible with the differential equation for $\1$
since both equations from system (\ref{system: variation of zero-point function and <vartheta>}) 
are derived from the same general formula in \cite[Lemma 6]{LN:2017}.
In particular, 
in the region where $x$ is large, 
it follows from step one that the differential equation for $\langle\vartheta(x)\rangle$  is equivalent 
to the differential equation for $\1$ in system (\ref{system: variation of zero-point function and A1 in terms of quotients of determinants}).
It remains to check the differential equation in the region where $x$ is small and the terms in $\langle\vartheta(x)\rangle$ that do not depend on $x$ dominate.
By the definition of $\langle\vartheta(x)\rangle$, we have 
\begin{displaymath}
\left(d-\frac{c}{8}\:\omega\right)\langle\vartheta(x)\rangle
%\=\left(d-\frac{c}{8}\omega\right)\left(-cx\1+\frac{\A }{4}\right)\\
=-cx\left(d+\frac{c}{24}\omega\right)\1
+\frac{1}{4}\left(d+\frac{c}{24}\omega\right)\A 
+x\frac{c^2}{6}\omega\1
-\frac{c}{24}\omega\A 
\:.
\end{displaymath} 
Using system (\ref{system: variation of zero-point function and A1 in terms of quotients of determinants}) 
and eqs (\ref{eq: quotient of determinants for Xi(3,k)}) and (\ref{eq: omega as difference of xi's over difference of Xi's}),
we obtain
\begin{equation}\label{eq: for comparison, the covariant derivative of <vartheta>}
\begin{split}
\left(d-\frac{c}{8}\:\omega\right)\langle\vartheta(x)\rangle
% \=\frac{c}{8}x\A \:\frac{\det\Xi_{3,0}}{\det V_3}
% +\frac{1}{4}C\:\frac{\det\Xi_{3,0}}{\det V_3}-\frac{1}{4}\A \:\frac{\det\Xi_{3,1}}{\det V_3}
% -\frac{c^2}{2}x\1\:\frac{\det\Xi_{3,1}}{\det V_3}
% +\frac{c}{8}\A \:\frac{\det\Xi_{3,1}}{\det V_3}\nn\\
% \=\frac{c}{4}\left(x\frac{\A }{2}-\frac{a}{15}\1\right)\:\frac{\det\Xi_{3,0}}{\det V_3}
% -\frac{1}{4}\left(2c^2x\1+\left(1-\frac{c}{2}\right)\A \right)\:\frac{\det\Xi_{3,1}}{\det V_3}\nn\\
\=\left(2c^2x\1+\left(1-\frac{c}{2}\right)\A \right)\:\sum_{s=1}^3\frac{\xi_sX_s}{p'(X_s)}\\
\-\left(cx\frac{\A }{2}-\frac{ca}{15}\1\right)\:\sum_{s=1}^3\frac{\xi_s}{p'(X_s)}
\:.
\end{split}
\end{equation} 
On the other hand, in the $(2,5)$ minimal model, we have by the proof of Proposition \ref{proposition: 2-point function of vartheta}.\ref{proposition: <vartheta vartheta>r}
for large $x$ and for $s=1,2,3$,
\begin{displaymath}%\label{eq. for 2-point function of vartheta when n=3}
\langle\vartheta(x)\vartheta(X_s)\rangle
=
-cx\langle\vartheta(X_s)\rangle
+\frac{c}{2}X_s^2\1
+\frac{8}{5}X_s\A 
+\frac{ca}{5}\1+O(x^{-1})
\:.
\end{displaymath}
We also note that by eqs (\ref{eq: quotient of determinants for Xi(3,2)}) and (\ref{eq: quotient of determinants for Xi(3,k)}),
\begin{displaymath}
\sum_{s=1}^3\frac{\xi_sX_s^2}{p'(X_s)}
% =-\frac{1}{4}\frac{\det\Xi_{3,2}}{\det V_3}
% =\frac{a}{12}\frac{\det\Xi_{3,0}}{\det V_3}
=-\frac{a}{3}\:\sum_{s=1}^3\frac{\xi_s}{p'(X_s)}
\:.
\end{displaymath}
So the term
\begin{displaymath}
2\sum_{s=1}^3
\frac{\xi_s}{p'(X_s)}\langle\vartheta(X_s)\vartheta(x)\rangle 
\end{displaymath}
% \begin{align*}
% 2\sum_{s=1}^3
% \frac{\xi_s}{p'(X_s)}\langle\vartheta(X_s)\vartheta(x)\rangle
% \=\left(2c^2x\1+\left(1-\frac{c}{2}\right)\A \right)\:\sum_{s=1}^3\frac{\xi_sX_s}{p'(X_s)}\\
% \+\left(-cx\frac{\A }{2}-\frac{ac}{3}\1+\frac{2ca}{5}\1+O(x^{-1})\right)\sum_{s=1}^3\frac{\xi_s}{p'(X_s)}
% \end{align*} 
matches the r.h.s.\ of eq.\ (\ref{eq: for comparison, the covariant derivative of <vartheta>}).
Since $\overline{\xi_1}=0$, we have for $m\geq 1$,
\begin{align*}
\sum_{s=1}^3\frac{\xi_s}{(x-X_s)^m}
\=O(x^{-(m+1)})\:.
\end{align*}
Thus by eqs (\ref{dp/p}) and (\ref{d(p'/p)}), the expression 
\begin{displaymath}
-\langle\vartheta(x)\rangle\frac{dp}{p}|_x
-\frac{c}{16}p'\:d\left(\frac{p'}{p}\right)|_x
\1 
\end{displaymath}
is $O(x^{-1})$.
This shows that the ODE for $\A$ and the ODE for $\langle\vartheta(x)\rangle$ are equivalent up to $O(1)$ terms.                                                     
From the general discussion \cite{LN:2017} we know that the respective r.h.s.\ of either differential equation has the correct singularities, 
so the remaining $O(x^{-1})$ terms must be zero.

\end{document}